\documentclass{amsart}
\usepackage{amsmath,amssymb,amsthm,amsfonts}
\usepackage{graphicx}
\usepackage{tikz}
\usepackage{geometry}
\usepackage{hyperref}
\usepackage{enumitem}
\usepackage{booktabs}
\usepackage{float}

\graphicspath{{./}}
\hypersetup{hidelinks}
\newtheorem{theorem}{Theorem}[section]
\newtheorem{lemma}[theorem]{Lemma}
\newtheorem{proposition}[theorem]{Proposition}
\newtheorem{corollary}[theorem]{Corollary}
\newtheorem{assumption}[theorem]{Assumption}
\theoremstyle{definition}
\newtheorem{definition}[theorem]{Definition}

\theoremstyle{remark}
\newtheorem{remark}[theorem]{Remark}

\input{tcilatex}

\begin{document}
\title[Exact two-sided p-values in NEFs]{Exact two-sided p-values in natural
exponential families: coincidence, non-uniqueness, and sample-size stability}
\author[S. K. Bar-Lev]{Shaul K. Bar-Lev}
\address{Faculty of Industrial Engineering and Technology Management,
HIT--Holon Institute of Technology, Holon, Israel}
\author[L. Hoessly]{Linard Hoessly}
\address{Data Center of the Swiss Transplant Cohort Study, University
Hospital Basel, Basel, 4031 Switzerland}
\email{linard.hoessly@hotmail.com}

\begin{abstract}
We study the non-uniqueness of exact two-sided $p$-values in continuous
one-parameter natural exponential families (NEFs). For directed one-sided
problems, the usual tail $p$-value agrees with the $p$-values obtained by
inverting UMP, UMPU, and likelihood-ratio (LR) tests. For a two-sided simple
null, we distinguish four constructions: equal-tail, density-ordered, UMPU,
and LR $p$-values. At a fixed null parameter, UMPU and equal-tail $p$-values
coincide if and only if the null law is symmetric about its mean; under a
regular two-branch density-level condition, the same fixed-null symmetry
characterization holds for UMPU versus density ordering and equal-tail
versus density ordering. Requiring any of these coincidences throughout the
NEF characterizes the Gaussian family. We combine these results with the
theorem of Bar-Lev, Bshouty and Letac that UMPU and LR $p$-values coincide
throughout a continuous NEF precisely for the normal, gamma and
inverse-Gaussian families.

We also investigate the two LR pairings not covered by those results. If
equal-tail and LR $p$-values coincide throughout a NEF satisfying our
standing regularity assumptions, then $(V^{2/3})^{\prime \prime \prime }=0$
on the mean domain; equivalently, $V^{2/3}$ is a quadratic polynomial. The
same coincidence also forces an explicit density-at-the-mean identity.
Within the power-variance class, these two necessary conditions, together
with the Lebesgue-dominated assumption, leave only the Gaussian family. For
an i.i.d.\ sample with canonical sufficient statistic $T_n=\sum_{i=1}^nX_i$,
persistence of equal-tail--LR coincidence throughout the family along an
unbounded sequence of sample sizes forces Gaussianity. A corresponding
density--LR statement is given conditionally on an explicitly stated
differentiated local Edgeworth expansion; we do not disguise that additional
asymptotic hypothesis as a consequence of the basic NEF assumptions.
Finally, inverse-Gaussian and hyperbolic-secant examples quantify numerical $%
p$-value differences, disagreement of rejection decisions, and differences
in power.
\end{abstract}

\maketitle
\tableofcontents

\section{Introduction}

{\ $p$-values are among the most widely used numerical summaries in
statistical inference. Their interpretation, limitations and practical use
have been discussed extensively; see, for example, the ASA statement of
Wasserstein and Lazar \cite{Wasserstein2016} and the subsequent discussion
in Wasserstein, Schirm and Lazar \cite{Wasserstein2019}. The present paper
concerns a narrower mathematical issue. Even when the null hypothesis is
simple, the model is continuous, and the resulting $p$-value is exact, a
two-sided $p$-value need not be a unique numerical object. }

For a one-sided problem, ``more extreme'' has a natural direction. In a
one-parameter exponential family with monotone likelihood ratio, this
direction is shared by the usual tail ordering, the UMP test and the
likelihood-ratio ordering. For a two-sided problem, especially under an
asymmetric null law, there is no single canonical way to order observations
on the two sides of the center. One may balance null tail probabilities,
order observations by null density, invert the UMPU tests, or order them by
generalized likelihood ratio. Each construction is mathematically natural,
but the induced rejection regions need not agree.

The ambiguity of two-sided $p$-values under asymmetry has a substantial
history. Gibbons and Pratt \cite{GibbonsPratt1975} discussed doubled-tail
and minimum-likelihood principles; George and Mudholkar \cite%
{GeorgeMudholkar1990} treated two-sided $p$-values for strictly unimodal
asymmetric null densities; Dunne, Pawitan and Doody \cite%
{DunnePawitanDoody1996} developed UMPU-based two-sided $p$-values for
discrete asymmetric models; and Mudholkar and Chaubey \cite%
{MudholkarChaubey2009} formulated an axiomatic link between valid $p$-values
and level-$\alpha$ tests. Berger and Delampady \cite{BergerDelampady1987}
and Kulinskaya \cite{Kulinskaya2008} provide further perspectives. Our
purpose is not to replace these constructions, but to place equal-tail,
density, UMPU and LR orderings in a common continuous-NEF framework and
determine when their induced exact $p$-value functions agree.

A central point of departure is the theorem of Bar-Lev, Bshouty and Letac 
\cite{BarLevBshoutyLetac2002}: under their regularity conditions, the
continuous NEFs for which the bilateral UMPU and generalized
likelihood-ratio tests coincide for every null parameter and every level are
precisely the normal, gamma and inverse-Gaussian families. That theorem is
an essential known ingredient rather than a new contribution of the present
paper. Our new fixed-null results concern the equal-tail and density
orderings, while our more structural new results concern the previously
untreated ET--LR edge and persistence under increasing sample size.

The equal-tail--LR edge is especially delicate. Equality of these two $p$%
-values at a null parameter forces the null mean to be a median, because the
LR $p$-value equals one at the null mean whereas the equal-tail $p$-value
equals one only at a median. This connects the problem with a difficult
characterization question. Letac, Mattner and Piccioni \cite%
{LetacMattnerPiccioni2018} proved a Gaussian characterization when the
natural parameter itself is a median of every exponential tilt and
explicitly isolated the corresponding mean--median question. Piccioni,
Kolodziejek and Letac \cite{PiccioniKolodziejekLetac2020} obtained related
fixed-quantile location and scale characterizations. We therefore do not
assert an unrestricted one-observation Gaussian theorem from a short
symmetry argument. Instead, we derive a new necessary variance-function
equation, solve the resulting problem in the power-variance class, and then
obtain a complete Gaussian characterization under sample-size stability.

It is important to distinguish our ET--LR differential equation from the
differential equations in Bar-Lev, Bshouty and Letac \cite%
{BarLevBshoutyLetac2002}. Their UMPU--GLR analysis ultimately yields the
factorized variance-function condition 
\begin{equation*}
V^{\prime }\bigl(3VV^{\prime \prime }-2(V^{\prime })^2\bigr) \bigl(%
2VV^{\prime \prime }-(V^{\prime })^2\bigr)=0,
\end{equation*}
and they also record the third-order equation 
\begin{equation*}
2(V^{\prime })^3-4VV^{\prime }V^{\prime \prime }+3V^2V^{\prime \prime \prime
}=0
\end{equation*}
arising from another intermediate relation. By contrast, ET--LR coincidence
below yields 
\begin{equation*}
4(V^{\prime })^3-9VV^{\prime }V^{\prime \prime }+9V^2V^{\prime \prime \prime
}=0,
\end{equation*}
equivalently $(V^{2/3})^{\prime \prime \prime }=0$. Thus the new equation
comes from a different coincidence relation and is not the UMPU--GLR
equation rewritten.

The contributions are therefore best separated by novelty. First, we give a
common exact-$p$-value framework and a direct nesting proof for the
continuous two-sided UMPU critical regions. Second, at a fixed null
parameter we prove that UMPU--ET coincidence is exactly symmetry about the
mean; under a regular two-branch density-level geometry, the analogous
statements hold for UMPU--density and ET--density. Requiring these
fixed-null symmetries throughout the NEF yields Gaussianity. Third, for
ET--LR coincidence throughout a sufficiently smooth NEF we derive the
variance-function restriction 
\begin{equation*}
9V^2V^{\prime \prime \prime }-9VV^{\prime }V^{\prime \prime }+4(V^{\prime
})^3=0,
\end{equation*}
or $(V^{2/3})^{\prime \prime \prime }=0$, together with the additional
density-at-the-mean identity in Proposition~\ref{prop:h-from-V}. These two
necessary conditions yield a complete Gaussian conclusion within the
continuous power-variance class. Fourth, for the canonical sufficient
statistic $T_n$ of an i.i.d.\ sample, we show that ET--LR coincidence
persisting along any unbounded sequence of sample sizes forces Gaussianity.
The density--LR analogue is stated only under an explicit differentiated
local Edgeworth condition. Fifth, we quantify the practical size of the
non-uniqueness in inverse-Gaussian and hyperbolic-secant examples.

The generic discrepancy bounds for two exact continuous $p$-values are
retained only as a calibration benchmark. They use uniform null marginals
rather than NEF structure; the distinctive mathematical content of the paper
lies in the coincidence characterizations, the ET--LR variance restriction,
and the sample-size results.

The paper is organized as follows. Section~2 defines the four $p$-value
constructions and records the NEF preliminaries. Section~3 gives the
fixed-null and global coincidence results, the finite-sample ET--LR
variance-function restriction, and the generic calibration benchmark.
Section~4 extends the theory to arbitrary i.i.d.\ sample size through the
canonical sufficient statistic. Section~5 gives explicit inverse-Gaussian
and hyperbolic-secant examples and decision-level comparisons. Section~6
discusses scope and open questions. The appendices contain the symmetry and
UMPU-nesting arguments, the positive-stable background, and the derivative
calculation underlying the ET--LR variance-function equation.

\subsection*{Acknowledgements}

We thank Samuel Pawel, Ioan Manolescu, Michael Dougoud, and Lucia de Andres
Bragado for helpful discussions and feedback.

\subsection*{AI use}
During the preparation of this manuscript, we used  GPT-4o
 for language edits as well as for latex support, and elicit.com for literature search. After
using them, we reviewed and edited the content as needed
and take full responsibility for its content.
\section{Definitions and preliminaries}

\subsection{Exact $p$-values from nested tests}

{\ Let $T=T(X)$ be a test statistic. A nested family of rejection regions $%
\{R_\alpha:0<\alpha<1\}$ satisfies 
\begin{equation*}
R_{\alpha_1}\subseteq R_{\alpha_2}\qquad (\alpha_1<\alpha_2).
\end{equation*}
If, for a null parameter set $\Theta_0$, 
\begin{equation*}
\sup_{\theta\in\Theta_0}\mathbb{P}_\theta(T\in R_\alpha)\leq\alpha,
\end{equation*}
the associated $p$-value is }

\begin{equation}  \label{eq:p-inversion}
{p(x)=\inf\big(\{\alpha\in(0,1):T(x)\in R_\alpha\}\cup\{1\}\big).}
\end{equation}

It is exact when 
\begin{equation*}
\sup_{\theta\in\Theta_0}\mathbb{P}_\theta\{p(X)\leq u\}=u, \qquad 0<u<1.
\end{equation*}
For a simple null $\Theta_0=\{\theta_0\}$ and a continuous $p(X)$, exactness
is equivalent to 
\begin{equation*}
p(X)\sim \mathrm{Unif}(0,1)\qquad\hbox{under }\theta_0.
\end{equation*}
See Lehmann and Romano \cite{LehmannRomano2022} for the standard
test-inversion framework. 

For directed problems we use $p_{\mathrm{us}}$ for the ordinary tail $p$%
-value. For two-sided problems we do not use the phrase ``usual $p$-value'',
because without an ordering convention the phrase ``as or more extreme'' is
not mathematically determinate under an asymmetric null. {\ For a
simple-versus-simple problem $H_0:\theta=\theta_0$ versus $%
H_1:\theta=\theta_1$, ``LR'' means the ordinary Neyman--Pearson likelihood
ratio $f_{\theta_0}(x)/f_{\theta_1}(x)$. For composite-null problems, and
for the two-sided simple-null problem that is the main subject of this
paper, we use the standard generalized likelihood-ratio statistic 
\begin{equation}  \label{eq:birkes}
\lambda_{\mathrm{GLR}}(x)= \frac{\sup_{\theta\in\Theta_0}f_\theta(x)} {%
\sup_{\theta\in\Theta}f_\theta(x)},
\end{equation}
with small values providing evidence against $H_0$. Thus the two directed
problems in Proposition~\ref{prop:one-sided} use their appropriate LR/GLR
form, while for $H_0:\theta=\theta_0$ versus $H_1:\theta\ne\theta_0$ the
statistic in \eqref{eq:birkes} is exactly the GLR ordering used throughout
the rest of the paper. Under monotone likelihood ratio, these directed
LR/GLR orderings are the same tail ordering; compare Birkes \cite{Birkes1990}%
.}

\begin{proposition}
\label{prop:one-sided} Let $\{f_\theta:\theta\in\Theta\}$ be a continuous
regular one-parameter exponential family with monotone likelihood ratio in
its canonical statistic $T$. Consider either 
\begin{equation*}
H_0:\theta=\theta_0\quad\hbox{versus}\quad H_1:\theta=\theta_1>\theta_0,
\end{equation*}
or 
\begin{equation*}
H_0:\theta\leq\theta_0\quad\hbox{versus}\quad H_1:\theta>\theta_0.
\end{equation*}
Then the directed tail, UMP, UMPU and likelihood-ratio $p$-values exist and
coincide: 
\begin{equation*}
p_{\mathrm{us}}=p_{\mathrm{UMP}}=p_{\mathrm{UMPU}}=p_{\mathrm{LR}}.
\end{equation*}
The analogous statement holds with all inequalities reversed.
\end{proposition}

\begin{proof}
Monotone likelihood ratio implies that the Neyman--Pearson rejection region
in the simple-versus-simple problem, and the UMP rejection region in the
one-sided composite problem, is an upper tail in $T$; see \cite%
{LehmannRomano2022}. The one-sided generalized likelihood-ratio region has
the same ordering. Inverting these nested upper-tail tests gives the null
tail probability. The UMP test is unbiased against the indicated one-sided
alternative, hence it is also UMP within the unbiased class. The lower-tail
case is identical after reversing the ordering.
\end{proof}

\subsection{Two-sided equal-tail and density-ordered $p$-values}

{\ Fix a simple null $H_0:\theta=\theta_0$. Assume $T$ has a continuous
strictly increasing null CDF $F_0$, density $f_0$, and quantile function $%
q_u=F_0^{-1}(u)$. }

\begin{definition}[Equal-tail $p$-value]
\label{def:ET} For $0<\alpha<1$, define 
\begin{equation*}
R^{\mathrm{ET}}_\alpha=(-\infty,q_{\alpha/2}]\cup[q_{1-\alpha/2},\infty).
\end{equation*}
The equal-tail $p$-value is the inversion of these regions. Equivalently, 
\begin{equation}  \label{eq:ET}
p_{\mathrm{ET}}(x)=2\min\{F_0(T(x)),1-F_0(T(x))\}.
\end{equation}
It is exact under $H_0$.
\end{definition}

\begin{definition}[Density-ordered $p$-value]
\label{def:dens} Assume $f_0(T)$ has a continuous distribution under $H_0$.
Define 
\begin{equation}  \label{eq:dens}
p_{\mathrm{dens}}(x) =\mathbb{P}_{\theta_0}\{f_0(T(X))\leq f_0(T(x))\}.
\end{equation}
Equivalently, observations are ordered by the null surprisal $-\log f_0(T)$.
This is the density or minimum-likelihood ordering discussed in the
precise-hypothesis literature; compare \cite%
{BergerDelampady1987,Kulinskaya2008}.
\end{definition}

\begin{lemma}
\label{lem:dens-exact} Under the assumptions of Definition~\ref{def:dens}, $%
p_{\mathrm{dens}}(X)\sim\mathrm{Unif}(0,1)$ under $H_0$.
\end{lemma}

\begin{proof}
Let $Y=f_0(T(X))$ and $G(y)=\mathbb{P}_0(Y\leq y)$. Then $p_{\mathrm{dens}%
}(X)=G(Y)$, which is uniform by the probability integral transform because $G
$ is continuous.
\end{proof}

\subsection{Natural exponential families}

{\ We use standard NEF notation; see Barndorff-Nielsen \cite%
{BarndorffNielsen1978}, Letac and Mora \cite{LetacMora1990}, and Bar-Lev and
Kokonendji \cite{BarLevKokonendji2017}. Let $\mu$ be a non-Dirac positive
Radon measure on $\mathbb{R}$ and define 
\begin{equation*}
L(\theta)=\int_{\mathbb{R}}e^{\theta x}\mu(dx), \qquad D=\{\theta\in\mathbb{R%
}:L(\theta)<\infty\}.
\end{equation*}
Assume 
\begin{equation*}
\Theta=\mathrm{int}\,D\neq\varnothing.
\end{equation*}
The cumulant transform is 
\begin{equation*}
k(\theta)=\log L(\theta),\qquad \theta\in\Theta,
\end{equation*}
and the generated NEF is 
\begin{equation}  \label{eq:NEF}
F(\mu)=\{P_\theta(dx)=e^{\theta x-k(\theta)}\mu(dx):\theta\in\Theta\}.
\end{equation}
On $\Theta$, $k$ is strictly convex and real analytic. Its first two
derivatives are 
\begin{equation*}
m(\theta)=k^{\prime }(\theta)=\mathbb{E}_\theta X, \qquad k^{\prime \prime
}(\theta)=\mathrm{Var}_\theta(X)>0.
\end{equation*}
The mean domain is 
\begin{equation*}
M=k^{\prime }(\Theta).
\end{equation*}
Writing $\psi=(k^{\prime})^{-1}:M\to\Theta$, the variance function is 
\begin{equation}  \label{eq:VF}
V(m)=k^{\prime \prime }(\psi(m))=\frac{1}{\psi^{\prime }(m)}.
\end{equation}
If $C$ denotes the convex support of $\mu$, the NEF is steep when 
\begin{equation*}
M=\mathrm{int}\,C.
\end{equation*}
This formulation is preferable here to requiring the effective domain $D$
itself to be open. For example, for a positive stable generating law one has 
$D=(-\infty,0]$ and $\Theta=(-\infty,0)$. }

\begin{assumption}
\label{ass:main} The NEF $F(\mu)$ is steep. The generating measure is
Lebesgue dominated, 
\begin{equation*}
\mu(dx)=h(x)\,dx,
\end{equation*}
where $h$ is strictly positive and twice continuously differentiable on an
open connected support interval $S$. Thus $S=\mathrm{int}\,C=M$.
\end{assumption}

{\ For $\theta_0\in\Theta$, write $m_0=k^{\prime }(\theta_0)$. Since the
family is steep and $M=S$, the unrestricted one-observation MLE at $x\in S$
is $\widehat\theta=\psi(x)$. The log generalized likelihood ratio for
testing $\theta=\theta_0$ is, apart from a monotone transformation, 
\begin{equation}  \label{eq:logLR}
r_{\theta_0}(x) =(\theta_0-\psi(x))x-k(\theta_0)+k(\psi(x)).
\end{equation}
Differentiation gives 
\begin{equation*}
r^{\prime }_{\theta_0}(x)=\theta_0-\psi(x), \qquad r^{\prime \prime
}_{\theta_0}(x)=-\psi^{\prime }(x)<0.
\end{equation*}
Hence $r_{\theta_0}$ is strictly concave, with its unique maximum at $x=m_0$%
. Its lower level sets are two-tailed nested rejection regions. }

The two-sided UMPU level-$\alpha$ test has a rejection region 
\begin{equation*}
R^{\mathrm{UMPU}}_\alpha=(-\infty,c_1(\alpha)]\cup[c_2(\alpha),\infty),
\qquad c_1(\alpha)<m_0<c_2(\alpha),
\end{equation*}
where the critical points solve the size and first-moment equations 
\begin{align}
F_0(c_1)+1-F_0(c_2)&=\alpha,  \label{eq:size} \\
\int_{-\infty}^{c_1}(x-m_0)f_0(x)\,dx
+\int_{c_2}^{\infty}(x-m_0)f_0(x)\,dx&=0.  \label{eq:unb}
\end{align}
For the one-parameter exponential-family problem $H_0:\theta=\theta_0$
against $\theta\neq\theta_0$, the standard UMPU theorem states that the
two-tail test satisfying \eqref{eq:size}--\eqref{eq:unb} is UMPU; see
Lehmann and Romano \cite[Section~4.2, especially equations (4.5)--(4.6)]%
{LehmannRomano2022}. Thus \eqref{eq:unb} is not being used as though local
unbiasedness by itself implied global unbiasedness. Appendix~A proves
directly that these critical regions exist uniquely and are nested in $%
\alpha $.

\begin{lemma}[Exactness of the four two-sided constructions]
\label{lem:four-exact} Under Assumption~\ref{ass:main}, the equal-tail, UMPU
and LR inversions are exact under every simple null. Whenever Definition~\ref%
{def:dens} applies, the density-ordered $p$-value is exact as well.
\end{lemma}

\begin{proof}
Exactness of $p_{\mathrm{ET}}$ follows directly from continuity of $F_0$,
and exactness of $p_{\mathrm{dens}}$ is Lemma~\ref{lem:dens-exact}. For
UMPU, Proposition~\ref{prop:nesting} gives nested rejection regions and %
\eqref{eq:size} gives null probability exactly $\alpha$ at every level. Test
inversion therefore gives 
\begin{equation*}
\mathbb{P}_{\theta_0}\{p_{\mathrm{UMPU}}(X)\leq\alpha\}=\alpha.
\end{equation*}

For LR, let 
\begin{equation*}
H_{\theta_0}(s)=\mathbb{P}_{\theta_0}\{r_{\theta_0}(X)\leq s\}.
\end{equation*}
Strict concavity of $r_{\theta_0}$ implies strict monotonicity on either
side of $m_0$. Since $f_0$ is positive and continuous on the connected
support, every nonmaximal level set of $r_{\theta_0}$ consists of at most
two points and has null probability zero. Moreover, every nonempty interval
of values in the interior of the range of $r_{\theta_0}$ has a preimage
containing a nonempty interval on at least one branch, hence positive
probability. Thus $H_{\theta_0}$ is continuous and strictly increasing on
the interior of the range, with endpoint values $0$ and $1$. The LR
inversion can be written 
\begin{equation*}
p_{\mathrm{LR}}(x)=H_{\theta_0}(r_{\theta_0}(x)).
\end{equation*}
The probability integral transform therefore gives $p_{\mathrm{LR}}(X)\sim%
\mathrm{Unif}(0,1)$.
\end{proof}

\begin{remark}[Meaning of affine equivalence]
\label{rem:affine} 
Let
\[
Y=aX+b,\qquad a\neq 0,
\]
and let $F(\nu)$ be the NEF obtained from $F(\mu)$ by this affine transformation. If
\[
X\sim P^\mu_\theta,
\]
then
\[
Y\sim P^\nu_{\theta/a}.
\]

The mean domain and variance function transform as
\[
M_\nu=aM_\mu+b,
\qquad
V_\nu(y)
=
a^2V_\mu\!\left(\frac{y-b}{a}\right).
\]
Hence Assumptions 2.5 and 3.1 are preserved.

The corresponding densities satisfy
\[
f^\nu_{\theta/a}(ax+b)
=
\frac{1}{|a|}f^\mu_\theta(x).
\]
Therefore affine transformations preserve all four $p$-value constructions. Density ordering is unchanged because all density values are multiplied by the same positive constant; the generalized likelihood ratio is unchanged because this constant cancels; equal-tail ordering is unchanged, with the two tails interchanged when $a<0$; and the UMPU size equation is unchanged while its centered first-moment equation is multiplied by $a$.

Thus, for
$
\bullet\in\{\mathrm{ET},\mathrm{dens},\mathrm{UMPU},\mathrm{LR}\}$,

\[
p^\nu_{\bullet,\theta/a}(ax+b)
=
p^\mu_{\bullet,\theta}(x).
\]

Consequently, any coincidence relation among the four $p$-values holds throughout $F(\mu)$ if and only if it holds throughout any affine transform of $F(\mu)$. This is what we mean by up to affine transformation.
\end{remark}
\section{Coincidence and non-coincidence of exact two-sided $p$-values}

\subsection{Fixed-null coincidence and global NEF consequences}

\begin{assumption}[Regular density-level geometry]
\label{ass:unimodal} For statements involving $p_{\mathrm{dens}}$, assume in
addition that for every $\theta\in\Theta$ the tilted density 
\begin{equation*}
f_\theta(x)=e^{\theta x-k(\theta)}h(x),\qquad x\in S,
\end{equation*}
has a unique interior mode $r_\theta$, satisfies 
\begin{equation*}
f_\theta^{\prime }(x)>0\quad(x<r_\theta), \qquad f_\theta^{\prime
}(x)<0\quad(x>r_\theta),
\end{equation*}
and tends to $0$ at the endpoints of $S$. Thus each nontrivial density level
has exactly two smooth branches, one on either side of the mode.
\end{assumption}

\begin{remark}
\label{rem:geom-assumption} Assumption~\ref{ass:unimodal} is a convenient
sufficient condition, not claimed to be minimal. The proofs use only the
corresponding geometry: every nontrivial density superlevel set is an
interval with two differentiable endpoints, and these intervals shrink to
the unique mode as the density level approaches its maximum. The derivative
sign condition is stated explicitly so that the inverse-function
differentiations are justified.
\end{remark}

\begin{theorem}[Fixed-null coincidence]
\label{thm:local} Suppose Assumption~\ref{ass:main} holds and fix $%
\theta_0\in\Theta$. All $p$-values in this theorem are computed under the
same simple null $H_0:\theta=\theta_0$, and equality means equality as
functions of the canonical observation.

\begin{enumerate}[ label=(\roman*)]

\item 
\begin{equation*}
p_{\mathrm{UMPU},\theta_0}=p_{\mathrm{ET},\theta_0}
\quad\Longleftrightarrow\quad P_{\theta_0}\text{ is symmetric about }%
m_0=k^{\prime }(\theta_0).
\end{equation*}
\end{enumerate}

If the density-level conditions of Assumption~\ref{ass:unimodal} hold for $%
f_{\theta_0}$, then

\begin{enumerate}[ label=(\roman*),resume]

\item 
\begin{equation*}
p_{\mathrm{UMPU},\theta_0}=p_{\mathrm{dens},\theta_0}
\quad\Longleftrightarrow\quad P_{\theta_0}\text{ is symmetric about }m_0;
\end{equation*}

\item 
\begin{equation*}
p_{\mathrm{ET},\theta_0}=p_{\mathrm{dens},\theta_0}
\quad\Longleftrightarrow\quad P_{\theta_0}\text{ is symmetric about }m_0.
\end{equation*}
\end{enumerate}
\end{theorem}

\begin{proof}
For (i), suppose first that $p_{\mathrm{UMPU},\theta_0}=p_{\mathrm{ET}%
,\theta_0}$. Equality of the $p$-value functions implies equality, up to
null boundaries, of their level-$\alpha$ rejection regions for every $\alpha$%
. Hence every equal-tail level-$\alpha$ test is UMPU and therefore unbiased.
Lemma~\ref{lem:ET-unbiased} implies that $P_{\theta_0}$ is symmetric about $%
m_0$. Conversely, if $P_{\theta_0}$ is symmetric about $m_0$, Lemma~\ref%
{lem:sym-UMPU} shows that the unique UMPU critical points are reflections of
one another about $m_0$. The two null tails are then equal, so the UMPU and
equal-tail regions, and hence their inversions, coincide.

For (ii), write $f=f_{\theta_0}$, $m=m_0$, and let $r$ be the unique mode.
For each $c\in(0,f(r))$, Assumption~\ref{ass:unimodal} gives unique points 
\begin{equation*}
a(c)<r<b(c),\qquad f(a(c))=f(b(c))=c.
\end{equation*}
Define 
\begin{equation*}
\alpha(c)=\mathbb{P}_{\theta_0}\{f(X)\leq c\} =F(a(c))+1-F(b(c)).
\end{equation*}
The two branch points vary continuously and strictly with $c$, so $\alpha(c)$
is continuous and strictly increasing from $0$ to $1$. Hence every
nontrivial density level corresponds to exactly one level $\alpha$ of the
density-ordered nested family. If $p_{\mathrm{UMPU},\theta_0}=p_{\mathrm{dens%
},\theta_0}$, the density-ordered region at level $\alpha(c)$ coincides with
the UMPU region at that level. The UMPU first-moment equation is therefore
equivalent to 
\begin{equation}  \label{eq:centralmoment-interval}
\int_{a(c)}^{b(c)}(x-m)f(x)\,dx=0 \qquad\hbox{for every }c\in(0,f(r)).
\end{equation}
Thus the conditional mean of $X$ given $a(c)<X<b(c)$ is $m$ for every $c$.
Since this conditional law is supported on $[a(c),b(c)]$, necessarily 
\begin{equation*}
a(c)<m<b(c).
\end{equation*}
As $c\uparrow f(r)$, both endpoints converge to $r$. The fixed point $m$
lies between them for every $c$, hence $m=r$.

By the inverse function theorem, 
\begin{equation*}
a^{\prime }(c)=\frac{1}{f^{\prime }(a(c))}, \qquad b^{\prime }(c)=\frac{1}{%
f^{\prime }(b(c))}.
\end{equation*}
Differentiating \eqref{eq:centralmoment-interval} and using $%
f(a(c))=f(b(c))=c$ yields 
\begin{equation*}
(b(c)-m)b^{\prime }(c)=(a(c)-m)a^{\prime }(c).
\end{equation*}
Consequently, 
\begin{equation*}
\frac{d}{dc}\{(b(c)-m)^2-(a(c)-m)^2\}=0.
\end{equation*}
Letting $c\uparrow f(m)$ shows that the constant is zero. Since $a(c)<m<b(c)$%
, 
\begin{equation*}
b(c)-m=m-a(c)
\end{equation*}
for every density level. Together with $f(a(c))=f(b(c))$, this proves
symmetry of $f$ about $m$.

Conversely, if $P_{\theta_0}$ is symmetric about $m_0$, Assumption~\ref%
{ass:unimodal} implies that each density level has reflected endpoints about 
$m_0$. Hence the density-ordered rejection regions are the equal-tail
regions; by part (i), these are also the UMPU regions.

For (iii), suppose $p_{\mathrm{ET},\theta_0}=p_{\mathrm{dens},\theta_0}$.
Equality of the level-$\alpha$ rejection regions implies, with $p=\alpha/2$, 
\begin{equation*}
f_0(q_p)=f_0(q_{1-p}),\qquad 0<p<\frac12.
\end{equation*}
Because $F_0$ is strictly increasing with positive differentiable density, 
\begin{equation*}
q^{\prime }_u=\frac{1}{f_0(q_u)}.
\end{equation*}
Therefore 
\begin{equation*}
\frac{d}{dp}\{q_p+q_{1-p}\} =\frac{1}{f_0(q_p)}-\frac{1}{f_0(q_{1-p})}=0.
\end{equation*}
Thus $q_p+q_{1-p}$ is constant. Letting $p\uparrow1/2$ gives 
\begin{equation*}
q_{1-p}=2q_{1/2}-q_p,
\end{equation*}
so $P_{\theta_0}$ is symmetric about its median. Its finite mean must equal
the symmetry center, hence the symmetry is about $m_0$. The converse follows
exactly as in part (ii).
\end{proof}

\begin{corollary}[Global Gaussian characterizations]
\label{cor:global} Suppose Assumption~\ref{ass:main} holds. Equality below
means equality of the corresponding $p$-value functions for every null
parameter $\theta_0\in\Theta$.

\begin{enumerate}[ label=(\roman*)]

\item $p_{\mathrm{UMPU}}=p_{\mathrm{ET}}$ throughout the NEF if and only if
the NEF is Gaussian up to affine transformation.
\end{enumerate}

If Assumption~\ref{ass:unimodal} holds for every $\theta\in\Theta$, then

\begin{enumerate}[ label=(\roman*),resume]

\item $p_{\mathrm{UMPU}}=p_{\mathrm{dens}}$ throughout the NEF if and only
if the NEF is Gaussian up to affine transformation;

\item $p_{\mathrm{ET}}=p_{\mathrm{dens}}$ throughout the NEF if and only if
the NEF is Gaussian up to affine transformation.
\end{enumerate}
\end{corollary}

\begin{proof}
By Theorem~\ref{thm:local}, any one of the stated global equalities forces
every $P_\theta$ to be symmetric about its mean. Lemma~\ref%
{lem:symmetric-NEF} then implies that $k$ is quadratic and the NEF is
Gaussian up to affine transformation. Conversely, every Gaussian tilt is
symmetric; its UMPU, equal-tail and density-ordered two-sided rejection
regions are the same symmetric tails.
\end{proof}

\begin{theorem}[Bar-Lev--Bshouty--Letac UMPU--LR characterization]
\label{thm:BLB} Under Assumption~\ref{ass:main}, 
\begin{equation*}
p_{\mathrm{UMPU}}=p_{\mathrm{LR}}\quad\hbox{throughout the NEF}
\end{equation*}
if and only if the NEF is, up to affine transformation, normal, gamma, or
inverse Gaussian.
\end{theorem}

\begin{proof}
Bar-Lev, Bshouty and Letac \cite{BarLevBshoutyLetac2002} assume that the NEF
is steep, that the generating measure is absolutely continuous, 
\begin{equation*}
\mu(dx)=h(x)\,dx,
\end{equation*}
and that $h$ is twice continuously differentiable on the mean domain. These
hypotheses are implied by Assumption~\ref{ass:main}: steepness is imposed
explicitly, $M=S=\mathrm{int}\,C$, and $h$ is strictly positive and $C^2$ on 
$S=M$. Their theorem therefore applies directly and yields precisely the
stated normal--gamma--inverse-Gaussian characterization.
\end{proof}

\begin{corollary}[Coincidence of all four constructions]
\label{cor:allfour} Under Assumptions~\ref{ass:main}--\ref{ass:unimodal}, 
\begin{equation*}
p_{\mathrm{ET}}=p_{\mathrm{dens}}=p_{\mathrm{UMPU}}=p_{\mathrm{LR}} \quad%
\hbox{throughout the family}
\end{equation*}
if and only if the NEF is Gaussian up to affine transformation.
\end{corollary}

\begin{proof}
If all four coincide, then in particular $p_{\mathrm{UMPU}}=p_{\mathrm{ET}}$
throughout the family, so Corollary~\ref{cor:global} gives Gaussianity.
Conversely, the Gaussian family satisfies Assumption~\ref{ass:unimodal};
symmetry makes the equal-tail, density and UMPU regions identical, and
Theorem~\ref{thm:BLB} gives equality with LR.
\end{proof}

\subsection{A coincidence graph}

{\ Figure~\ref{fig:graph} summarizes the complete one-observation global
characterizations that require no additional power-variance or
sample-size-stability restriction. The two absent LR pairings are treated
separately below: ET--LR is solved within the power-variance class, and both
missing LR edges receive sample-size-stable statements. }

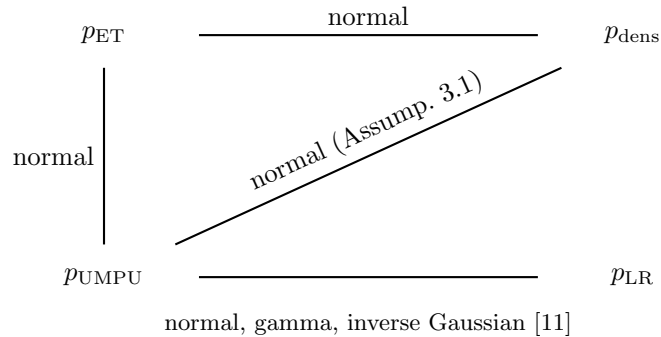
\begin{figure}[htbp]
\centering
\begin{tikzpicture}
\node[rounded corners,minimum width=2.5cm,minimum height=.85cm] (ET) at (0,3.2) {$p_{\rm ET}$};
\node[rounded corners,minimum width=2.5cm,minimum height=.85cm] (DEN) at (7.0,3.2) {$p_{\rm dens}$};
\node[rounded corners,minimum width=2.5cm,minimum height=.85cm] (U) at (0,0) {$p_{\rm UMPU}$};
\node[rounded corners,minimum width=2.5cm,minimum height=.85cm] (LR) at (7.0,0) {$p_{\rm LR}$};

\draw[thick] (ET)--node[left]{normal} (U);
\draw[thick] (ET)--node[above]{normal} (DEN);
\draw[thick] (DEN)--node[pos=.48,above,sloped]{normal (Assump.~\ref{ass:unimodal})} (U);
\draw[thick] (U)--(LR);
\node[font=\small] at (3.5,-0.62) {normal, gamma, inverse Gaussian \cite{BarLevBshoutyLetac2002}};
\end{tikzpicture}
{}
\caption{Global proved coincidence structure among the four exact two-sided $%
p$-value constructions. The three symmetry-based edges are obtained from
Theorem~\protect\ref{thm:local} and Corollary~\protect\ref{cor:global}; the
UMPU--LR edge is the theorem of Bar-Lev, Bshouty and Letac. The absent
ET--LR and density--LR edges are not asserted as unrestricted
one-observation characterizations.}
\label{fig:graph}
\end{figure}

\begin{remark}
\label{rem:LRedges} Theorem~\ref{thm:BLB} does not classify the ET--LR or
density--LR edges, because LR tests need not be unbiased. The next
subsection gives a finite-sample necessary differential equation for ET--LR
and a complete consequence within the power-variance class. A full
one-observation classification for an arbitrary NEF is deliberately not
claimed.
\end{remark}

\subsection{A finite-sample restriction for the equal-tail--LR edge}

{Throughout this subsection, primes on the variance function denote
derivatives with respect to the mean parameter $m$. No smoothness beyond
Assumption~\ref{ass:main} is imposed. The ET--LR hypothesis upgrades the
standing $C^2$ regularity of $h$ to real analyticity by itself. This
bootstrap is carried out in Stage~2 of the proof of Theorem~\ref%
{thm:ETLR-ODE} and recorded as Proposition~\ref{prop:h-from-V}.}

\begin{theorem}[A necessary quadratic-root variance-function restriction for
ET--LR coincidence]
\label{thm:ETLR-ODE} {Suppose Assumption~\ref{ass:main} holds}. If 
\begin{equation*}
p_{\mathrm{ET},\theta}=p_{\mathrm{LR},\theta} \qquad\hbox{for every }%
\theta\in\Theta
\end{equation*}
as functions of the canonical observation, then:

\begin{enumerate}[ label=(\roman*)]

\item the mean $m(\theta)$ is a median of $P_\theta$ for every $\theta$;

\item on the mean domain $M$, 
\begin{equation}  \label{eq:ETLR-ODE}
9V(m)^2V^{\prime \prime \prime }(m)-9V(m)V^{\prime }(m)V^{\prime \prime
}(m)+4\{V^{\prime }(m)\}^{3}=0;
\end{equation}

\item equivalently, 
\begin{equation}  \label{eq:V23quad}
\frac{d^3}{dm^3}V(m)^{2/3}=0,
\end{equation}
so there exist real constants $a,b,c$ such that 
\begin{equation}  \label{eq:Vquadraticroot}
V(m)=\{am^2+bm+c\}^{3/2},\qquad m\in M,
\end{equation}
with $am^2+bm+c>0$ on $M$.
\end{enumerate}
\end{theorem}

\begin{proof}
Fix $\theta$ and write $m=k^{\prime }(\theta)$, $F=F_\theta$, $f=f_\theta$,
and $q(u)=F^{-1}(u)$. Let 
\begin{equation*}
r_m(x) =(\theta-\psi(x))x-k(\theta)+k(\psi(x))
\end{equation*}
be the log generalized likelihood ratio in the canonical coordinate. It is
strictly concave and has its unique maximum at $x=m$. Consequently 
\begin{equation*}
p_{\mathrm{LR}}(m)=1.
\end{equation*}
If $p_{\mathrm{ET}}=p_{\mathrm{LR}}$, then $p_{\mathrm{ET}}(m)=1$, {so that $%
2\min\{F(m),1-F(m)\}=1$ and hence directly} 
\begin{equation*}
F(m)=\frac12.
\end{equation*}
This proves (i) and gives $q(1/2)=m$.

{The remainder of the proof proceeds in two stages. Stage~1 uses only
Assumption~\ref{ass:main} and produces the third-order identity \eqref{eq:l1}%
; Stage~2 uses that identity to upgrade the regularity of $h$, after which
the fifth-order computation is legitimate. In particular no smoothness
hypothesis beyond Assumption~\ref{ass:main} is required anywhere. }

\emph{Stage 1.} By Assumption~\ref{ass:main}, $f=e^{\theta x-k(\theta)}h(x)$
is $C^2$ and strictly positive on $S$, hence $F$ is $C^3$ with nonvanishing
derivative, and the inverse-function theorem gives $q\in C^3(0,1)$. The map $%
\psi=(k^{\prime })^{-1}$ is real analytic on $M$, because $k$ is real
analytic and strictly convex on $\Theta$, so $r_m$ is real analytic in its
observation coordinate. Thus 
\begin{equation*}
G_m(u)=r_m(q(u))
\end{equation*}
is $C^3$ near $u=1/2$, which is all that is needed to form $G_m^{\prime
\prime \prime }(1/2)$.

Define 
\begin{equation*}
H_m(s)=\mathbb{P}_\theta\{r_m(X)\leq s\}.
\end{equation*}
By the same two-branch argument as in Lemma~\ref{lem:four-exact}, $H_m$ is
continuous and strictly increasing on the interior of the range of $r_m$,
and 
\begin{equation*}
p_{\mathrm{LR}}(x)=H_m(r_m(x)).
\end{equation*}
Since 
\begin{equation*}
p_{\mathrm{ET}}(q(u)) =2\min\{u,1-u\} =p_{\mathrm{ET}}(q(1-u)),
\end{equation*}
ET--LR equality and strict monotonicity of $H_m$ imply 
\begin{equation}  \label{eq:G-sym}
G_m(u)=G_m(1-u),\qquad 0<u<1.
\end{equation}
In particular, every odd derivative of $G_m$ that exists at $u=1/2$
vanishes; in Stage~1 this is used for $G_m^{\prime\prime\prime}(1/2)=0$, and
in Stage~2 for $G_m^{(5)}(1/2)=0$.

For brevity write $f(m)=f_{\psi(m)}(m)$. Put 
\begin{equation*}
\ell_j(m)= \left. \frac{\partial^j}{\partial x^j}\log f_{\psi(m)}(x)
\right|_{x=m}.
\end{equation*}
Using 
\begin{equation*}
r_m^{\prime \prime }(m)=-\frac1{V(m)},\qquad r_m^{\prime \prime \prime }(m)=%
\frac{V^{\prime }(m)}{V(m)^2}, \qquad q^{\prime }(1/2)=\frac1{f(m)},
\end{equation*}
the identity $G_m^{\prime \prime \prime }(1/2)=0$ gives 
\begin{equation}  \label{eq:l1}
\ell_1(m)=-\frac{V^{\prime }(m)}{3V(m)}.
\end{equation}

{\emph{Stage 2.} Since $\theta\in\Theta$ was arbitrary and $%
m=k^{\prime}(\theta)$ maps $\Theta$ bijectively onto $M$, \eqref{eq:l1}
holds for every $m\in M$. Because $\ell_1(m)=\psi(m)+(\log h)^{\prime }(m)$,
it can be rewritten as 
\begin{equation}  \label{eq:logh-prime}
(\log h)^{\prime }(m)=-\frac{V^{\prime }(m)}{3V(m)}-\psi(m),\qquad m\in M=S,
\end{equation}
whose right-hand side is real analytic on $M$, since $V$ and $\psi$ are.
Hence $\log h$, and therefore $h$, is real analytic on $S$; consequently $f$%
, $F$ and $q$ are real analytic and $G_m\in C^\infty$ near $u=1/2$. The
fifth-order computation below is therefore legitimate under Assumption~\ref%
{ass:main} alone. The argument is sequential and not circular: \eqref{eq:l1}
was obtained in Stage~1 from $C^2$ regularity only. The conclusion %
\eqref{eq:logh-prime} is of independent interest and is recorded as
Proposition~\ref{prop:h-from-V}.}

Since 
\begin{equation*}
\ell_1(m)=\psi(m)+(\log h)^{\prime }(m), \qquad \psi^{\prime
}(m)=\frac1{V(m)},
\end{equation*}
differentiating \eqref{eq:l1} along the mean parameter expresses $\ell_2$
and $\ell_3$ through $V,V^{\prime },V^{\prime \prime },V^{\prime \prime
\prime }$. Substitution into the fifth-derivative identity displayed in
Appendix~\ref{app:ETLRcalc} yields 
\begin{equation}  \label{eq:G5factor}
G_m^{(5)}(1/2) = -\frac{2}{27 f(m)^5V(m)^4} \left[ 9V(m)^2V^{\prime \prime
\prime }(m)-9V(m)V^{\prime }(m)V^{\prime \prime }(m)+4\{V^{\prime }(m)\}^{3}%
\right].
\end{equation}
Equation \eqref{eq:G-sym} makes the left side zero, proving %
\eqref{eq:ETLR-ODE}. Finally, 
\begin{equation}  \label{eq:V23identity}
\frac{d^3}{dm^3}V^{2/3} = \frac{2}{27V^{7/3}} \left(9V^2V^{\prime \prime
\prime }-9VV^{\prime }V^{\prime \prime }+4(V^{\prime })^3\right),
\end{equation}
so \eqref{eq:ETLR-ODE} is equivalent to \eqref{eq:V23quad}, and $V^{2/3}$ is
a polynomial of degree at most two.
\end{proof}

\begin{proposition}[ET--LR coincidence determines the density at the mean]
\label{prop:h-from-V} Suppose Assumption~\ref{ass:main} holds and 
\begin{equation*}
p_{\mathrm{ET},\theta}=p_{\mathrm{LR},\theta} \qquad\hbox{for every }%
\theta\in\Theta.
\end{equation*}
Then there is a constant $C>0$ such that 
\begin{equation}  \label{eq:f-at-mean}
f_{\psi(m)}(m)=C\,V(m)^{-1/3},\qquad m\in M;
\end{equation}
that is, the null density evaluated at its own mean is proportional to $%
V^{-1/3}$. Equivalently, 
\begin{equation}  \label{eq:h-from-V}
h(x)=C\,V(x)^{-1/3}\exp\{k(\psi(x))-x\psi(x)\},\qquad x\in S.
\end{equation}
In particular $h$ is real analytic on $S$.
\end{proposition}

\begin{proof}
Identity \eqref{eq:logh-prime} holds on $M=S$. Since 
\begin{equation*}
\frac{d}{dm}\{m\psi(m)-k(\psi(m))\} =\psi(m)+m\psi^{\prime }(m)-k^{\prime
}(\psi(m))\psi^{\prime }(m) =\psi(m),
\end{equation*}
integration of \eqref{eq:logh-prime} gives \eqref{eq:h-from-V}. Substituting %
\eqref{eq:h-from-V} into $f_\theta(x)=e^{\theta x-k(\theta)}h(x)$ at $%
\theta=\psi(m)$ and $x=m$, the two exponential factors cancel and %
\eqref{eq:f-at-mean} follows. Real analyticity was established in Stage~2 of
the proof of Theorem~\ref{thm:ETLR-ODE}.
\end{proof}

\begin{remark}[Scope and strength of \eqref{eq:f-at-mean}]
\label{rem:h-from-V} Three comments.

(i) A generating measure is determined by the NEF it generates only up to $%
h\mapsto Ce^{ax}h$, so \eqref{eq:h-from-V} must be, and is, invariant under
that ambiguity: replacing $\mu$ by $e^{ax}\mu$ sends $\theta\mapsto\theta-a$%
, hence $\psi\mapsto\psi-a$, while $(\log h)^{\prime }\mapsto(\log
h)^{\prime }+a$, and the two changes cancel in \eqref{eq:logh-prime}. Form %
\eqref{eq:f-at-mean} makes the invariance manifest, since it refers only to
members of the family and to $V$.

(ii) Adding \eqref{eq:f-at-mean} to the differential equation %
\eqref{eq:ETLR-ODE} gives a strictly stronger condition than %
\eqref{eq:ETLR-ODE} alone. For the Gaussian family $V\equiv\sigma^2$ and $%
f_{\psi(m)}(m)=(2\pi\sigma^2)^{-1/2}$, both sides constant, as they must be.
By contrast, for the gamma NEF of shape $p>0$ one has $V(m)=m^2/p$ and $%
f_{\psi(m)}(m)=p^pe^{-p}\Gamma(p)^{-1}m^{-1}$, whereas \eqref{eq:f-at-mean}
would require $f_{\psi(m)}(m)\propto m^{-2/3}$; the mismatch is by the
factor $m^{-1/3}$ for every shape. For the inverse-Gaussian NEF generated by
the L\'evy density \eqref{eq:Levy} one has $V(m)=2m^3$ and $\psi(m)=-1/(4m^2)
$, so \eqref{eq:h-from-V} would force $h(x)\propto x^{-1}e^{-1/(4x)}$,
whereas the true generating density is $\propto x^{-3/2}e^{-1/(4x)}$: the
exponential factors agree and the powers do not.

(iii) Proposition~\ref{prop:h-from-V} therefore excludes the gamma and
inverse-Gaussian families from ET--LR coincidence analytically, without
recourse to the numerical comparison of Section~\ref{sec:examples}. It is
also the mechanism behind Remark~\ref{rem:quadratic-root-not-sufficient}
below, where it disposes of a variance function that satisfies %
\eqref{eq:ETLR-ODE} identically.
\end{remark}

\begin{corollary}[ET--LR coincidence in the power-variance class]
\label{cor:ETLR-PVF} Suppose Assumption~\ref{ass:main} holds, and suppose
that $\mathcal{F}(\mu)$ has a power variance function: after an affine
change of coordinate, exactly one
of the following holds for some $A>0$, 
\begin{equation*}
\hbox{(a)}\quad p=0,\qquad M=\mathbb{R},\qquad V(m)=A,
\end{equation*}
\begin{equation*}
\hbox{(b)}\quad p\neq 0,\qquad M=(b,\infty),\qquad V(m)=A(m-b)^p.
\end{equation*}
Then 
\begin{equation*}
p_{\mathrm{ET}}=p_{\mathrm{LR}}\quad\hbox{throughout the NEF}
\end{equation*}
if and only if case \textup{(a)} occurs, that is, if and only if the NEF is
Gaussian up to affine transformation.
\end{corollary}

\begin{remark}
\label{rem:pvf_cases} The two cases in Corollary~\ref{cor:ETLR-PVF} must be
listed separately. In case~(a), the variance function is constant and its
maximal mean domain is $\mathbb{R}$. In case~(b), the maximal mean domain is
a half-line $(b,\infty)$. Since no affine transformation maps $\mathbb{R}$
onto a half-line, the Gaussian case cannot be included in the class
described by case~(b).
\end{remark}

\begin{proof}
By Remark~\ref{rem:affine} the hypothesis is unaffected by the affine change
of coordinate in the statement, so we may assume $V$ is already in the form
\textup{(a)} or \textup{(b)}.
In case \textup{(a)}, $k^{\prime \prime }\equiv A$, so $k$ is quadratic and
the NEF is Gaussian up to affine transformation; each $P_\theta $ is then
symmetric about its mean, the LR level sets are the equal-tail regions, and $%
p_{\mathrm{ET}}=p_{\mathrm{LR}}$.

In case \textup{(b)}, suppose $p_{\mathrm{ET}}=p_{\mathrm{LR}}$ throughout
the family. By \eqref{eq:V23quad}, $V^{2/3}=A^{2/3}(m-b)^{2p/3}$ has
vanishing third derivative on $(b,\infty)$, that is 
\begin{equation}  \label{eq:PVFfactor}
\frac{4}{27}\,A^{2/3}\,p\,(2p-3)(p-3)\,(m-b)^{(2p-9)/3}=0,\qquad m>b.
\end{equation}
The factor $(m-b)^{(2p-9)/3}$ is strictly positive there, so the admissible
exponents are $p=0$, $p=3/2$ and $p=3$. The first is excluded by the
hypothesis of case \textup{(b)}. For the remaining two we invoke the
standard power-variance classification of Bar-Lev and Enis \cite%
{BarLevEnis1986}, in the exponential-dispersion formulation of J\o rgensen 
\cite{Jorgensen1987}. For a maximal half-line mean domain, the exponent $%
p=3/2\in(1,2)$ belongs to the compound-Poisson--gamma range and therefore
has an atom at the boundary $b$; it is excluded by the purely
Lebesgue-dominated Assumption~\ref{ass:main}. The exponent $p=3$ is the
inverse-Gaussian case. It satisfies the differential restriction %
\eqref{eq:ETLR-ODE}, so the differential equation alone does not exclude it.
The additional necessary identity of Proposition~\ref{prop:h-from-V} does:
Remark~\ref{rem:h-from-V}(ii) shows that ET--LR coincidence would force a
generating density proportional to $x^{-1}e^{-1/(4x)}$ rather than the
correct $x^{-3/2}e^{-1/(4x)}$. Section~\ref{sec:examples} confirms this
numerically with a null member for which $p_{\mathrm{LR}}(m)=1$ but $p_{%
\mathrm{ET}}(m)<1$. Hence case \textup{(b)} admits no NEF with ET--LR
coincidence, and the equivalence follows.
\end{proof}

\begin{remark}[The quadratic-root restriction is necessary but not
sufficient: the normal inverse-Gaussian NEF]
\label{rem:quadratic-root-not-sufficient} Equation~\eqref{eq:Vquadraticroot}
is necessary but is not a one-observation classification, as the following
counterexample shows.

Let $k(\theta)=1-\sqrt{1-\theta^{2}}$ on $\Theta=(-1,1)$. Then $m=k^{\prime
}(\theta)=\theta(1-\theta^{2})^{-1/2}$, so $M=\mathbb{R}$, $%
1+m^{2}=(1-\theta^{2})^{-1}$ and 
\begin{equation}  \label{eq:NIGvf}
V(m)=k^{\prime \prime }(\psi(m))=(1+m^{2})^{3/2},\qquad
\psi(m)=m(1+m^{2})^{-1/2}.
\end{equation}
This is the symmetric normal inverse-Gaussian NEF of Barndorff-Nielsen \cite%
{BarndorffNielsen1997} with $\alpha=\delta=1$, $\beta=\mu=0$, generated by
the strictly positive real analytic Lebesgue density 
\begin{equation}  \label{eq:NIGdensity}
h(x)=\frac{e}{\pi}\,\frac{K_{1}\big(\sqrt{1+x^{2}}\big)}{\sqrt{1+x^{2}}}%
,\qquad x\in\mathbb{R},
\end{equation}
$K_{1}$ being the modified Bessel function of the second kind. The family is
steep with $M=S=\mathbb{R}$, so Assumption~\ref{ass:main} holds, and $%
V^{2/3}=1+m^{2}$ is quadratic, so \eqref{eq:ETLR-ODE} and %
\eqref{eq:Vquadraticroot} hold identically. Thus $(1+m^{2})^{3/2}$ is not
merely an algebraically admissible solution but a genuine variance function
of a continuous steep NEF.

It does not, however, have ET--LR coincidence. We test the invariant form %
\eqref{eq:f-at-mean}, so that the ambiguity of Remark~\ref{rem:h-from-V}(i)
plays no role. Since $k(\psi(m))-m\psi(m)=1-\sqrt{1+m^{2}}$, the exponential
factors in $f_{\psi(m)}(m)=e^{m\psi(m)-k(\psi(m))}h(m)$ combine with %
\eqref{eq:NIGdensity} to give, with $z=\sqrt{1+m^{2}}$, 
\begin{equation}  \label{eq:NIGfatmean}
f_{\psi(m)}(m)=\frac{1}{\pi}\,\frac{e^{z}K_{1}(z)}{z},\qquad V(m)^{-1/3}=%
\frac{1}{z}.
\end{equation}
So \eqref{eq:f-at-mean} would require $e^{z}K_{1}(z)$ to be constant on $%
[1,\infty)$, whereas $e\,K_{1}(1)>1$ and $e^{z}K_{1}(z)=\sqrt{\pi/(2z)}%
\{1+O(z^{-1})\}\to0$. Hence $p_{\mathrm{ET},\theta}\neq p_{\mathrm{LR}%
,\theta}$ for some $\theta$.

Consequently \eqref{eq:Vquadraticroot} is strictly weaker than ET--LR
coincidence even within the positive analytic solutions that are genuine
variance functions of continuous steep NEFs, and Proposition~\ref%
{prop:h-from-V} supplies what is needed to exclude this family. Theorem~\ref%
{thm:ETLR-ODE} identifies a narrow candidate class but asserts no
sufficiency outside the power-variance subclass.
\end{remark}

\begin{remark}[Why the unrestricted one-observation ET--LR edge is delicate]

\label{rem:meanmedian} The first conclusion of Theorem~\ref{thm:ETLR-ODE}
already shows that a general ET--LR characterization contains a mean--median
characterization problem. Letac, Mattner and Piccioni \cite%
{LetacMattnerPiccioni2018} proved a related Gaussian characterization and
explicitly identified the mean--median problem as delicate. Our ET--LR
hypothesis is stronger than mean--median coincidence, as witnessed by the
additional constraints \eqref{eq:ETLR-ODE} and \eqref{eq:f-at-mean}, but we
do not claim that these conditions alone solve the unrestricted NEF problem. 
{Remark~\ref{rem:quadratic-root-not-sufficient} shows that %
\eqref{eq:ETLR-ODE} by itself certainly does not, and identifies %
\eqref{eq:f-at-mean} as the additional ingredient that disposes of the
normal inverse-Gaussian family.}
\end{remark}

\subsection{Exactness does not imply closeness}

\begin{proposition}
\label{prop:distance} Let $p_1(X)$ and $p_2(X)$ be two continuous exact $p$%
-values under the same simple null. Then 
\begin{equation}  \label{eq:L1bound}
\mathbb{E}_0|p_1(X)-p_2(X)|\leq\frac12,
\end{equation}
and 
\begin{equation}  \label{eq:L2bound}
\mathbb{E}_0\{p_1(X)-p_2(X)\}^2\leq\frac13.
\end{equation}
Both bounds are sharp. On the other hand, exactness alone gives no
nontrivial universal pointwise bound: the supremum of $|p_1-p_2|$ can equal $%
1$.
\end{proposition}

\begin{proof}
Put $U=p_1(X)$ and $W=p_2(X)$. Exactness gives uniform marginals. Since 
\begin{equation*}
\mathbb{E}|U-W|=\mathbb{E}(U+W)-2\mathbb{E}\min(U,W)=1-2\mathbb{E}\min(U,W),
\end{equation*}
and 
\begin{equation*}
\mathbb{P}\{\min(U,W)>t\} =\mathbb{P}(U>t,W>t) \geq\max(1-2t,0),
\end{equation*}
we have 
\begin{equation*}
\mathbb{E}\min(U,W) \geq\int_0^{1/2}(1-2t)\,dt =\frac14,
\end{equation*}
which proves \eqref{eq:L1bound}. Further, 
\begin{equation*}
0\leq\mathbb{E}(U+W-1)^2 =2\mathbb{E}(U^2)+2\mathbb{E}(UW)-2 =-\frac13+2%
\mathbb{E}(UW),
\end{equation*}
so $\mathbb{E}(UW)\geq1/6$. Therefore 
\begin{equation*}
\mathbb{E}(U-W)^2 =\frac23-2\mathbb{E}(UW) \leq\frac13.
\end{equation*}
For $W=1-U$, equality holds in both bounds, and 
\begin{equation*}
\sup_{0<u<1}|u-(1-u)|=1.
\end{equation*}
\end{proof}

\begin{remark}
\label{rem:distance-secondary} Proposition~\ref{prop:distance} is
distribution-free and uses only the fact that the two null marginals are
uniform. It is included as a calibration benchmark, not as an NEF
characterization. For the examples below we therefore also report
model-specific discrepancies and decision disagreement probabilities.
\end{remark}

\section{Arbitrary sample size and the canonical sufficient statistic}

\label{sec:samplesize}

{\ Let $X_1,\ldots,X_n$ be i.i.d.\ from the NEF \eqref{eq:NEF}, and put 
\begin{equation*}
T_n=\sum_{i=1}^nX_i,\qquad \overline X_n=T_n/n.
\end{equation*}
The statistic $T_n$ is canonical and sufficient for $\theta$. Its
distribution is the NEF generated by the convolution measure $\mu^{*n}$: 
\begin{equation}  \label{eq:TnNEF}
P_{\theta}^{(n)}(dt) = \exp\{\theta t-nk(\theta)\}\,\mu^{*n}(dt).
\end{equation}
Thus 
\begin{equation*}
k_n(\theta)=nk(\theta),\qquad \mathbb{E}_\theta T_n=nm(\theta),\qquad 
\mathrm{Var}_\theta(T_n)=nV(m(\theta)).
\end{equation*}
The mean domain is $nM$, and, if $t$ denotes the mean coordinate for $T_n$,
its variance function is 
\begin{equation}  \label{eq:Vn}
V_n(t)=nV(t/n).
\end{equation}
}

\begin{remark}[Relation with the reproducibility results of Bar-Lev and Enis]

\label{rem:reproducibility} Equation~\eqref{eq:Vn} is a universal
convolution identity for NEFs and does not, by itself, assert that the
distribution of a scaled sum belongs to the same one-parameter family as the
summands. This distinction is closely related to the reproducibility problem
studied by Bar-Lev and Enis \cite{BarLevEnis1986}. In their terminology,
reproducibility asks for constants $\alpha_n$ and a reparametrization $g_n$
such that 
\begin{equation*}
\mathcal{L}\!\left(\alpha_n\sum_{i=1}^n X_i\right)=P_{g_n(\theta)} \quad%
\hbox{whenever }X_i\overset{\mathrm{i.i.d.}}{\sim}P_\theta,
\end{equation*}
so that, after deterministic rescaling, the convolution returns to the
original NEF. Bar-Lev and Enis demonstrated the intimate connection between
this stronger property and NEFs with power variance functions, and treated
convolution as one of the central structural features of that class.

This is directly relevant to the perspective of Proposition~\ref%
{prop:finite-n}: the normal, gamma, and inverse-Gaussian NEFs correspond to
the power-variance exponents $0$, $2$, and $3$, respectively, and their
scaled-convolution behavior is part of the reproducibility structure.
Nevertheless, Proposition~\ref{prop:finite-n} does not require the
reproducibility theorem. For a fixed $n$ it uses only the universal
identities $k_n=nk$, $M_n=nM$, and $V_n(t)=nV(t/n)$, together with the
normal--gamma--inverse-Gaussian characterization. Thus the connection with 
\cite{BarLevEnis1986} is conceptual and structural, while the proof below is
strictly more elementary and does not invoke reproducibility as an
additional assumption.
\end{remark}

\begin{remark}[Which density is being ordered?]
\label{rem:Tn-density} For a sample of size $n$, $p_{\mathrm{dens},n}$ in
this paper means density ordering of the one-dimensional sufficient
statistic $T_n$ under its marginal null density, equivalently of $\overline
X_n$ because the two coordinates differ affinely. It does not mean ordering
the full sample vector by its $n$-dimensional joint density. Indeed, 
\begin{equation*}
f_\theta^{\otimes n}(x_1,\ldots,x_n) = \exp\{\theta T_n-nk(\theta)\}
\prod_{i=1}^n h(x_i),
\end{equation*}
and the factor $\prod_i h(x_i)$ is generally not determined by $T_n$.
\end{remark}

\begin{proposition}[Finite-$n$ transfer of the coincidence results]
\label{prop:finite-n} Fix $n\geq1$ and suppose the marginal law of $T_n$
satisfies the regularity assumptions required for the corresponding
one-dimensional results.

\begin{enumerate}[ label=(\roman*)]

\item 
\begin{equation*}
p_{\mathrm{UMPU},n}=p_{\mathrm{ET},n} \quad\hbox{throughout }\theta
\end{equation*}
if and only if the original NEF is Gaussian up to affine transformation.

\item If the marginal density of $T_n$ satisfies the two-branch geometry of
Assumption~\ref{ass:unimodal}, then the same Gaussian characterization holds
for 
\begin{equation*}
p_{\mathrm{UMPU},n}=p_{\mathrm{dens},n} \quad\hbox{and}\quad p_{\mathrm{ET}%
,n}=p_{\mathrm{dens},n}.
\end{equation*}

\item 
\begin{equation*}
p_{\mathrm{UMPU},n}=p_{\mathrm{LR},n} \quad\hbox{throughout }\theta
\end{equation*}
if and only if the original NEF is, up to affine transformation, normal,
gamma, or inverse Gaussian.
\end{enumerate}
\end{proposition}

\begin{proof}
Apply Corollary~\ref{cor:global} and Theorem~\ref{thm:BLB} to the NEF
generated by $\mu^{*n}$. If $nk(\theta)$ is quadratic, then $k(\theta)$ is
quadratic, proving (i), and the density pairings follow identically under
the stated marginal geometry.

For (iii), one must preserve the specific affine-power variance forms, not
merely polynomial degree. If the $T_n$-NEF is Gaussian, then $%
V_n(t)\equiv\sigma_n^2$, so from \eqref{eq:Vn}, 
\begin{equation*}
V(m)=\frac{\sigma_n^2}{n},
\end{equation*}
which is constant. If the $T_n$-NEF is an affine gamma family, then for
constants $A_n>0$ and $b_n$, 
\begin{equation*}
V_n(t)=A_n(t-b_n)^2.
\end{equation*}
Putting $t=nm$ in \eqref{eq:Vn} gives 
\begin{equation*}
V(m)=nA_n\left(m-\frac{b_n}{n}\right)^2,
\end{equation*}
the affine gamma form. Likewise, if 
\begin{equation*}
V_n(t)=B_n(t-b_n)^3
\end{equation*}
is the affine inverse-Gaussian form, then 
\begin{equation*}
V(m)=n^2B_n\left(m-\frac{b_n}{n}\right)^3.
\end{equation*}
Conversely, substituting a constant, shifted-square, or shifted-cube form
for $V$ into \eqref{eq:Vn} produces the same corresponding form for $V_n$.
Moreover, the mean domain of the $T_n$-NEF is exactly $nM$, so the full
variance-function pair $(V,M)$, not merely the polynomial degree, is carried
back and forth by this scaling. Since an NEF is determined by its variance
function together with its maximal mean domain, the normal--gamma--inverse-
Gaussian triple is preserved in both directions.
\end{proof}

\begin{theorem}[Sample-size-stable mean--median characterization]
\label{thm:meanmedian-samplesize} Suppose Assumption~\ref{ass:main} holds,
and let $\mathcal{N}\subset\mathbb{N}$ be unbounded. If 
\begin{equation}  \label{eq:meanmedian-asymp}
\mathbb{P}_\theta(T_n\leq nm(\theta))-\frac12 =o(n^{-1/2}) \qquad \hbox{as }%
n\to\infty\hbox{ through }\mathcal{N},
\end{equation}
for every $\theta\in\Theta$, then the NEF is Gaussian up to affine
transformation. In particular the conclusion holds if $nm(\theta)$ is a
median of $T_n$ for every $\theta\in\Theta$ and every $n\in\mathcal{N}$.
Conversely, in a Gaussian NEF $nm(\theta)$ is a median of $T_n$ for every $%
\theta$ and every $n$.
\end{theorem}

\begin{proof}
The converse is immediate from symmetry of the Gaussian law of $T_n$ about $%
nm(\theta)$. For necessity, fix $\theta$ and write $m=k^{\prime }(\theta)$, $%
\sigma^2=k^{\prime \prime }(\theta)$ and $\gamma_1(\theta)=k^{\prime \prime
\prime }(\theta)/\sigma^3$. Set 
\begin{equation*}
Z_{n,\theta}=\frac{T_n-nm}{\sigma\sqrt n},
\end{equation*}
so that \eqref{eq:meanmedian-asymp} reads $\mathbb{P}_\theta(Z_{n,\theta}%
\leq0)-\frac12=o(n^{-1/2})$.

Because $\theta$ lies in the interior of the natural parameter space, the
tilted law has a moment generating function in a neighbourhood of zero and
hence moments of every order. By Assumption~\ref{ass:main} it has an $L^1$
density, so the Riemann--Lebesgue lemma gives $|\varphi_\theta(t)|\to0$ as $%
|t|\to\infty$, which implies Cram\'er's strong nonlattice condition $%
\limsup_{|t|\to\infty}|\varphi_\theta(t)|<1$. The one-term Edgeworth
expansion for i.i.d.\ sums under these moment and Cram\'er conditions, as in
Petrov \cite[Chapter~VI, Theorem~7]{Petrov1975}, yields 
\begin{equation}  \label{eq:EdgeworthCDF}
\mathbb{P}_\theta(Z_{n,\theta}\leq z) = \Phi(z) + \frac{\gamma_1(\theta)}{%
6\sqrt n}(1-z^2)\phi(z) +o(n^{-1/2}),
\end{equation}
uniformly in $z$. Evaluating at $z=0$ and combining with %
\eqref{eq:meanmedian-asymp}, 
\begin{equation*}
o(n^{-1/2}) = \mathbb{P}_\theta(Z_{n,\theta}\leq0)-\frac12 = \frac{%
\gamma_1(\theta)\phi(0)}{6\sqrt n} +o(n^{-1/2}) \qquad(n\in\mathcal{N}).
\end{equation*}
Multiplying by $\sqrt n$ and letting $n\to\infty$ through $\mathcal{N}$
yields $\gamma_1(\theta)=0$, hence $k^{\prime \prime \prime }(\theta)=0$.
Since $\theta\in\Theta$ was arbitrary, $k$ is quadratic and the NEF is
Gaussian.
\end{proof}

\begin{remark}
\label{rem:meanmedian-scope} Theorem~\ref{thm:meanmedian-samplesize} is what
the sample-size argument actually proves: no property of the ET or LR
orderings is used beyond the single consequence that the mean of $T_n$ is a
median. Stating it separately also makes the connection with the
mean--median problem of Letac, Mattner and Piccioni \cite%
{LetacMattnerPiccioni2018} explicit, and shows that exact equality of mean
and median is more than is needed --- the $o(n^{-1/2})$ rate in %
\eqref{eq:meanmedian-asymp} suffices.
\end{remark}

\begin{theorem}[Sample-size-stable ET--LR coincidence]
\label{thm:sample-ETLR} Suppose Assumption~\ref{ass:main} holds, and let $%
\mathcal{N}\subset\mathbb{N}$ be unbounded. For every $n\in\mathcal{N}$,
define $p_{\mathrm{ET},n,\theta}$ and $p_{\mathrm{LR},n,\theta}$ from the
null distribution of $T_n$. Then 
\begin{equation*}
p_{\mathrm{ET},n,\theta}=p_{\mathrm{LR},n,\theta} \quad \hbox{for every }%
\theta\in\Theta\hbox{ and every }n\in\mathcal{N }
\end{equation*}
if and only if the NEF is Gaussian up to affine transformation.
\end{theorem}

\begin{proof}
{The Gaussian converse is immediate. For necessity, write $m=k^{\prime
}(\theta)$. The LR statistic based on $T_n$ has its unique maximum at $T_n=nm
$, so $p_{\mathrm{LR},n,\theta}(nm)=1$; ET--LR equality therefore gives $p_{%
\mathrm{ET},n,\theta}(nm)=1$, that is 
\begin{equation*}
\mathbb{P}_\theta(T_n\leq nm)=\frac12 \qquad(n\in\mathcal{N}),
\end{equation*}
for every $\theta\in\Theta$. This is hypothesis \eqref{eq:meanmedian-asymp}
of Theorem~\ref{thm:meanmedian-samplesize}, which yields the conclusion.}
\end{proof}

\begin{assumption}[Differentiated local Edgeworth condition]
\label{ass:density-edgeworth} For the density--LR sample-size statement, fix 
$\theta$ and let $g_{n,\theta}$ denote the density of $Z_{n,\theta}$. Assume
that for the sample sizes under consideration the marginal density ordering
of $T_n$ is exact in the sense of Definition~\ref{def:dens}, that $%
g_{n,\theta}$ is differentiable near zero, and that 
\begin{equation}  \label{eq:EdgeworthDensityDeriv}
g_{n,\theta}^{\prime }(0) = -\frac{\gamma_1(\theta)}{2\sqrt n}\phi(0)
+o(n^{-1/2})
\end{equation}
as $n\to\infty$ along the specified sequence. This is the derivative at zero
of the first local Edgeworth density expansion. It is stronger than the
distribution-function expansion used in Theorem~\ref{thm:sample-ETLR}; local
Edgeworth theory gives primitive sufficient smoothness and
Fourier-integrability conditions, but we state the exact asymptotic
condition needed below rather than claiming that it follows from Assumption~%
\ref{ass:main} alone.
\end{assumption}

\begin{corollary}[Conditional sample-size-stable density--LR coincidence]
\label{thm:sample-densLR} Suppose Assumptions~\ref{ass:main} and~\ref%
{ass:density-edgeworth} hold, and let $\mathcal{N}\subset\mathbb{N}$ be
unbounded. If 
\begin{equation*}
p_{\mathrm{dens},n,\theta}=p_{\mathrm{LR},n,\theta} \quad \hbox{for every }%
\theta\in\Theta\hbox{ and every }n\in\mathcal{N},
\end{equation*}
where density ordering is the marginal ordering of $T_n$ specified in Remark~%
\ref{rem:Tn-density}, then the NEF is Gaussian up to affine transformation.
The converse holds.
\end{corollary}

\begin{proof}
Fix $\theta$. At $T_n=nm$ the LR $p$-value equals one. Equality therefore
gives 
\begin{equation*}
p_{\mathrm{dens},n,\theta}(nm)=1.
\end{equation*}
Thus the marginal null density of $T_n$ attains its global maximum at $nm$.
Since the density is differentiable there and $nm$ is an interior point, 
\begin{equation*}
g_{n,\theta}^{\prime }(0)=0 \qquad(n\in\mathcal{N})
\end{equation*}
on the standardized scale. Combining this with %
\eqref{eq:EdgeworthDensityDeriv}, multiplying by $\sqrt n$, and letting $%
n\to\infty$ through $\mathcal{N}$ gives 
\begin{equation*}
\gamma_1(\theta)=0.
\end{equation*}
Thus $k^{\prime \prime \prime }(\theta)=0$ for every $\theta$, and the NEF
is Gaussian. Conversely, for a Gaussian sample, the sufficient-statistic
density is symmetric and all four orderings coincide for every $n$.
\end{proof}

\begin{remark}
\label{rem:sample-strength} Theorem~\ref{thm:sample-ETLR} is an
unconditional sample-size-stability result under Assumption~\ref{ass:main}.
Corollary~\ref{thm:sample-densLR} is intentionally conditional on the
differentiated local Edgeworth expansion \eqref{eq:EdgeworthDensityDeriv}.
Neither result converts the unresolved unrestricted $n=1$ ET--LR or
density--LR problem into a claim without proof.
\end{remark}

\section{Explicit examples and decision consequences}

\label{sec:examples}

\subsection{The L\'evy/inverse-Gaussian case}

{\ The positive $1/2$-stable density is 
\begin{equation}  \label{eq:Levy}
g_{1/2}(x)=\frac{1}{2\sqrt\pi}x^{-3/2}\exp\left(-\frac{1}{4x}\right), \qquad
x>0,
\end{equation}
with Laplace transform 
\begin{equation*}
\int_0^\infty e^{\theta x}g_{1/2}(x)\,dx =\exp\{-\sqrt{-\theta}\}, \qquad
\theta<0.
\end{equation*}
The corresponding NEF density is 
\begin{equation*}
f_\theta(x) =\exp\{\theta x+\sqrt{-\theta}\}g_{1/2}(x).
\end{equation*}
It is an inverse-Gaussian family with 
\begin{equation*}
\mu=\frac{1}{2\sqrt{-\theta}}, \qquad \lambda=\frac12.
\end{equation*}
In the conventional $IG(\mu,\lambda)$ parameterization, 
\begin{equation*}
f(x)=\left(\frac{\lambda}{2\pi x^3}\right)^{1/2} \exp\left\{-\frac{%
\lambda(x-\mu)^2}{2\mu^2x}\right\}, \qquad x>0.
\end{equation*}
Its CDF is 
\begin{equation}  \label{eq:IGcdf}
F_{\mu,\lambda}(x) =\Phi\left(\sqrt{\frac{\lambda}{x}}\left(\frac{x}{\mu}%
-1\right)\right) +\exp\left(\frac{2\lambda}{\mu}\right) \Phi\left(-\sqrt{%
\frac{\lambda}{x}}\left(\frac{x}{\mu}+1\right)\right).
\end{equation}
Therefore 
\begin{equation*}
p_{\mathrm{ET}}(x)=2\min\{F_{\mu,\lambda}(x),1-F_{\mu,\lambda}(x)\}.
\end{equation*}
}

For the L\'evy-generated case $\lambda=1/2$, the log generalized likelihood
ratio against the unrestricted mean is 
\begin{equation*}
\log\Lambda_\mu(x) =-\frac{(x-\mu)^2}{4\mu^2x}.
\end{equation*}
The point on the opposite side of $\mu$ having the same LR statistic is 
\begin{equation*}
x^*=\frac{\mu^2}{x}.
\end{equation*}
Since the inverse-Gaussian family belongs to the Bar-Lev--Bshouty--Letac
triple, 
\begin{equation*}
p_{\mathrm{UMPU}}=p_{\mathrm{LR}},
\end{equation*}
and hence 
\begin{equation}  \label{eq:IGlr}
p_{\mathrm{LR}}(x)=p_{\mathrm{UMPU}}(x)
=F_{\mu,1/2}\!\left(\min\{x,\mu^2/x\}\right)
+1-F_{\mu,1/2}\!\left(\max\{x,\mu^2/x\}\right).
\end{equation}

The numerical difference is already substantial at $\mu=1$. For example, 
\begin{equation*}
x=0.5:\qquad p_{\mathrm{ET}}\approx0.9803, \qquad p_{\mathrm{LR}}=p_{\mathrm{%
UMPU}}\approx0.6171.
\end{equation*}
At the null mean $x=1$, 
\begin{equation*}
p_{\mathrm{LR}}(1)=p_{\mathrm{UMPU}}(1)=1, \qquad p_{\mathrm{ET}%
}(1)\approx0.5724,
\end{equation*}
because the asymmetric inverse-Gaussian null has mean different from its
median. Thus exactness of both constructions does not imply numerical
agreement.

\subsection{The hyperbolic-secant NEF}

{\ A particularly convenient smooth steep example outside the
normal--gamma--inverse-Gaussian triple is obtained from the standard
hyperbolic-secant density 
\begin{equation}  \label{eq:HSbase}
h(x)=\frac{1}{2\cosh(\pi x/2)}, \qquad x\in\mathbb{R}.
\end{equation}
Its moment generating function is $\sec\theta$ for $|\theta|<\pi/2$; see,
for example, Fischer \cite{Fischer2014}. Hence the NEF has 
\begin{equation*}
\Theta=\left(-\frac\pi2,\frac\pi2\right), \qquad k(\theta)=-\log\cos\theta,
\end{equation*}
\begin{equation*}
m(\theta)=\tan\theta, \qquad V(m)=1+m^2.
\end{equation*}
Because $m(\Theta)=\mathbb{R}$, the family is steep. }

There is also a convenient exact CDF representation. Put 
\begin{equation*}
U=\frac{e^{\pi X}}{1+e^{\pi X}}.
\end{equation*}
A direct change of variables shows that under $P_\theta$, 
\begin{equation*}
U\sim \mathrm{Beta}\left(\frac12+\frac{\theta}{\pi}, \frac12-\frac{\theta}{%
\pi}\right).
\end{equation*}
Consequently, 
\begin{equation}  \label{eq:HScdf}
F_\theta(x) =I_{\func{logistic}(\pi x)} \left(\frac12+\frac{\theta}{\pi},
\frac12-\frac{\theta}{\pi}\right),
\end{equation}
which is evaluated in base R by combining \texttt{pbeta} and \texttt{plogis} 
\cite{RCore2026}.

The tilted density is 
\begin{equation*}
f_\theta(x) =\frac{\cos\theta\,e^{\theta x}}{2\cosh(\pi x/2)}.
\end{equation*}
Its unique mode is 
\begin{equation}  \label{eq:HSmode}
r_\theta =\frac{2}{\pi}\func{arctanh}\left(\frac{2\theta}{\pi}\right).
\end{equation}
Thus Assumption~\ref{ass:unimodal} holds.

All four $p$-values can now be computed without a specialized package. The
equal-tail value follows directly from \eqref{eq:HScdf}. For the density
value, if $x$ is on one side of the mode, solve on the opposite side for $y_{%
\mathrm{dens}}$ satisfying 
\begin{equation*}
f_\theta(y_{\mathrm{dens}})=f_\theta(x),
\end{equation*}
and then set 
\begin{equation*}
p_{\mathrm{dens}}(x) =F_\theta(\min\{x,y_{\mathrm{dens}}\})
+1-F_\theta(\max\{x,y_{\mathrm{dens}}\}).
\end{equation*}
For the UMPU value define the centered partial moment 
\begin{equation}  \label{eq:Hcentered}
H_\theta(z)=\int_{-\infty}^z(t-m(\theta))f_\theta(t)\,dt.
\end{equation}
The opposite UMPU boundary $y_{\mathrm{UMPU}}$ is the solution on the other
side of $m(\theta)$ of 
\begin{equation*}
H_\theta(y_{\mathrm{UMPU}})=H_\theta(x),
\end{equation*}
which is exactly the unbiasedness equation. Finally, since $\psi(x)=\arctan
x $, the log LR statistic is 
\begin{equation}  \label{eq:HSlr}
r_\theta(x) =(\theta-\arctan x)x +\log\frac{\cos\theta}{\cos(\arctan x)}.
\end{equation}
The opposite LR boundary solves $r_\theta(y_{\mathrm{LR}})=r_\theta(x)$.

Figure~\ref{fig:HSplot} uses $\theta_0=\pi/4$. Then 
\begin{equation*}
m_0=1, \qquad r_{\theta_0}=\frac{2}{\pi}\func{arctanh}(1/2) \approx0.349699,
\end{equation*}
and the transformed beta parameters are $(3/4,1/4)$. The four exact $p$%
-values are visibly different.

\begin{figure}[htbp]
\centering
\includegraphics[width=.88%
\textwidth,draft=false]{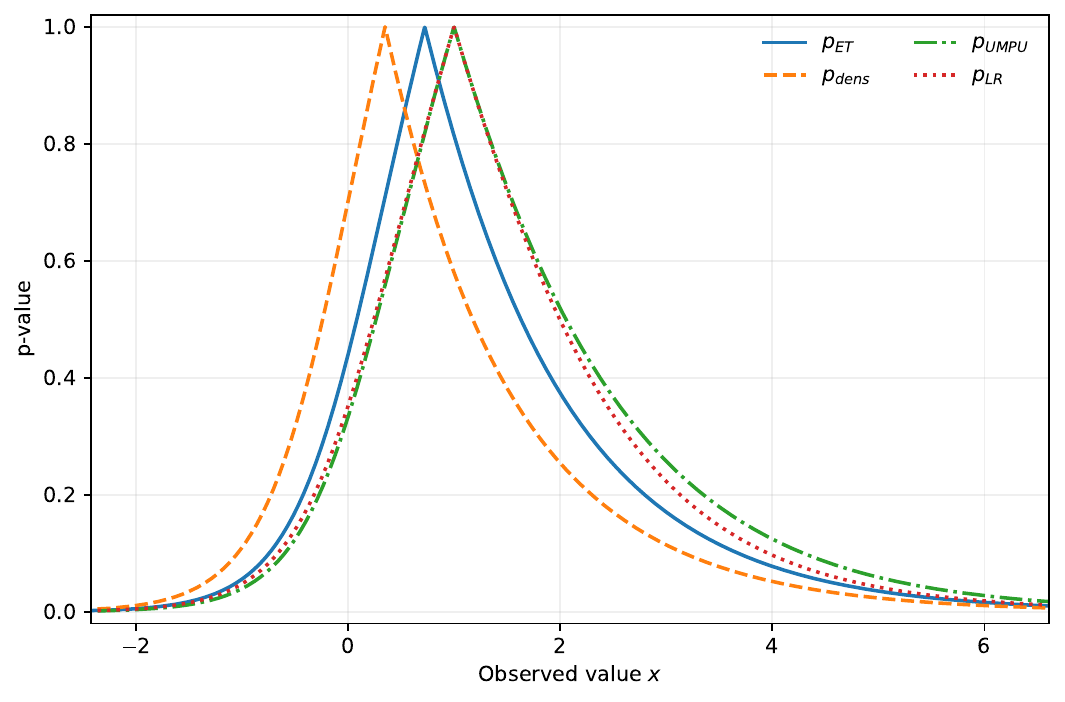}
\caption{Equal-tail, density-ordered, UMPU and likelihood-ratio $p$-values
in the hyperbolic-secant NEF at $\protect\theta_0=\protect\pi/4$. The
supplementary base-R file \texttt{hyperbolic\_secant\_pvalues.R} reproduces
the curves, critical intervals, discrepancy calculations,
decision-disagreement probabilities, and power table.}
\label{fig:HSplot}
\end{figure}

{\ On a dense grid spanning null quantiles $0.0005$ through $0.9995 $, the
largest observed pairwise differences are approximately 
\begin{equation*}
\begin{array}{c|c}
\text{pair} & \text{largest grid difference} \\ \hline
p_{\mathrm{ET}},p_{\mathrm{dens}} & 0.291 \\ 
p_{\mathrm{ET}},p_{\mathrm{UMPU}} & 0.186 \\ 
p_{\mathrm{ET}},p_{\mathrm{LR}} & 0.186 \\ 
p_{\mathrm{dens}},p_{\mathrm{UMPU}} & 0.442 \\ 
p_{\mathrm{dens}},p_{\mathrm{LR}} & 0.431 \\ 
p_{\mathrm{UMPU}},p_{\mathrm{LR}} & 0.034%
\end{array}%
\end{equation*}
These are grid-based descriptive values, not asserted analytic suprema.
Their purpose is to show that the choice of exact two-sided ordering can be
practically material even in a smooth one-parameter NEF. }

For reproducibility, the supplementary R calculation uses 800 equally spaced
probability points between null quantiles $0.0005$ and $0.9995$. Root
finding is performed with \texttt{uniroot}; numerical integration for the
UMPU centered moment uses relative tolerance $10^{-10}$, and the main root
tolerances are $10^{-10}$, with $10^{-8}$ for the nested UMPU root because
each evaluation itself contains a numerical integral. The numerical values
reported below are outputs of that script rather than inputs to the
analytical arguments.

The differences are also visible at the level of testing decisions. For two
procedures $i$ and $j$, define the null disagreement probability at level $%
\alpha$ by 
\begin{equation}  \label{eq:decision-disagreement}
D_{ij}(\theta_0;\alpha) = \mathbb{P}_{\theta_0} \left( \mathbf{1}%
\{p_i(X)\leq\alpha\} \neq \mathbf{1}\{p_j(X)\leq\alpha\} \right).
\end{equation}
At $\theta_0=\pi/4$ and $\alpha=0.05$, the numerically determined critical
intervals are 
\begin{equation*}
\begin{array}{c|cc}
\text{procedure} & \text{left critical point} & \text{right critical point}
\\ \hline
\mathrm{ET} & -1.048123 & 4.563125 \\ 
\mathrm{dens} & -1.343658 & 4.049527 \\ 
\mathrm{UMPU} & -0.900121 & 5.223057 \\ 
\mathrm{LR} & -0.980326 & 4.798729%
\end{array}%
\end{equation*}
and the corresponding null disagreement probabilities are 
\begin{equation*}
\begin{array}{c|c}
\text{pair} & D_{ij}(\theta_0;0.05) \\ \hline
\mathrm{ET,\ dens} & 0.02484 \\ 
\mathrm{ET,\ UMPU} & 0.02022 \\ 
\mathrm{ET,\ LR} & 0.00845 \\ 
\mathrm{dens,\ UMPU} & 0.04507 \\ 
\mathrm{dens,\ LR} & 0.03329 \\ 
\mathrm{UMPU,\ LR} & 0.01178%
\end{array}%
\end{equation*}
Thus, although every procedure has exact null rejection probability $0.05$,
density ordering and UMPU make different decisions on about $4.5\%$ of null
observations in this example.

With the null fixed at $\theta_0=\pi/4$ and size $0.05$, representative
rejection probabilities under alternatives are 
\begin{equation*}
\begin{array}{c|cccc}
\theta & \mathrm{ET} & \mathrm{dens} & \mathrm{UMPU} & \mathrm{LR} \\ \hline
0 & 0.1217 & 0.0779 & 0.1520 & 0.1348 \\ 
\pi/6 & 0.0523 & 0.0366 & 0.0648 & 0.0575 \\ 
\pi/4 & 0.0500 & 0.0500 & 0.0500 & 0.0500 \\ 
\pi/3 & 0.0996 & 0.1202 & 0.0796 & 0.0918 \\ 
1.25 & 0.2331 & 0.2706 & 0.1926 & 0.2177%
\end{array}%
\end{equation*}
The UMPU test has power at least its size on both sides by construction. The
other exact tests need not be unbiased. {The table also shows that, in this
example, none of ET, density ordering, or LR is unbiased. Indeed, at $%
\theta=\pi/3$ each has strictly larger power than the UMPU test while all
have the same null size $0.05$. Since the UMPU test is uniformly most
powerful within the class of unbiased level-$0.05$ tests, any competing test
with larger power at even one alternative cannot itself belong to that
unbiased class.} These numerical examples are illustrations rather than
evidence for the analytical theorems.

\section{Discussion}

{\ The one-sided and two-sided situations should be sharply distinguished.
In a continuous one-parameter exponential family with monotone likelihood
ratio, a directed alternative determines a natural ordering, and tail, UMP,
UMPU and LR inversion agree. For a two-sided simple null, there is no
corresponding universal ordering under asymmetry. }

The first structural contribution is the local-to-global separation in
Theorem~\ref{thm:local} and Corollary~\ref{cor:global}. At a fixed null
parameter, coincidence of UMPU with equal-tail inversion is exactly a
symmetry statement about that null law; under the regular density-level
geometry, the same is true for the two density-related pairings. Requiring
the same local symmetry at every exponential tilt forces a quadratic
cumulant and hence the Gaussian NEF. This is distinct from the previously
known theorem of Bar-Lev, Bshouty and Letac \cite{BarLevBshoutyLetac2002},
where UMPU and LR inversion select the larger
normal--gamma--inverse-Gaussian triple.

The second contribution concerns the LR edges that are not covered by that
earlier characterization. We deliberately do not assert a complete
unrestricted one-observation ET--LR or density--LR theorem. For ET--LR,
Theorem~\ref{thm:ETLR-ODE} gives a new necessary variance-function equation, 
\begin{equation*}
9V^2V^{\prime \prime \prime }-9VV^{\prime }V^{\prime \prime }+4(V^{\prime
})^3=0,
\end{equation*}
equivalently $(V^{2/3})^{\prime \prime \prime }=0$. This equation is not the
differential equation arising in the Bar-Lev--Bshouty--Letac UMPU--LR proof:
their analysis leads to different variance-function identities associated
with unbiasedness of the LR regions. The present equation is obtained
instead from the quantile symmetry forced by ET--LR equality. Corollary~\ref%
{cor:ETLR-PVF} gives a complete conclusion in the power-variance subclass,
while Remark~\ref{rem:meanmedian} explains why we do not promote this to an
unrestricted one-observation characterization without further work.

The third contribution is the sample-size formulation. The canonical
statistic for an i.i.d. sample is $T_n=\sum_iX_i$, whose cumulant is $nk$
and whose variance function is $V_n(t)=nV(t/n)$. Proposition~\ref%
{prop:finite-n} transfers the established UMPU comparisons to any fixed
sample size. Its proof uses the specific constant, shifted-square and
shifted-cube variance-function forms of the normal, gamma and
inverse-Gaussian families; polynomial degree alone would not suffice. More
importantly, Theorem~\ref{thm:sample-ETLR} shows that persistence of ET--LR
coincidence along any unbounded sequence of sample sizes forces zero
skewness at every tilt and hence Gaussianity.

The density--LR sample-size statement is intentionally weaker in its
hypotheses. Corollary~\ref{thm:sample-densLR} is conditional on the
differentiated local Edgeworth expansion \eqref{eq:EdgeworthDensityDeriv}.
We do not present that expansion as a consequence of Assumption~\ref%
{ass:main} alone. Under primitive Fourier-integrability and smoothness
conditions it follows from standard local Edgeworth theory, but the precise
weakest conditions are not part of the present paper. Stating the expansion
explicitly prevents the asymptotic regularity from being hidden inside an
informal phrase such as ``standard smoothness assumptions.''

The distinction in Remark~\ref{rem:Tn-density} is essential. Density
ordering of the one-dimensional sufficient statistic $T_n$ is not the same
as ranking the full sample by the joint null density. The former is the
natural extension of the canonical-coordinate density $p$-value studied here
and is invariant under replacing $T_n$ by $\overline X_n$; the latter
depends on ancillary sample configuration through $\prod_i h(X_i)$ and is a
different procedure.

The numerical examples sharpen the practical message. Reporting merely ``the
two-sided $p$-value'' can conceal a genuine choice of ordering. In the
hyperbolic-secant example, all four constructions are exact at the null, yet
the level-$0.05$ rejection regions differ enough that some pairs disagree on
several percent of null observations, and the corresponding power functions
differ as well. The supplementary R script reproduces the plotted curves,
grid discrepancies, critical intervals, disagreement probabilities and power
values reported in Section~\ref{sec:examples}.

Proposition~\ref{prop:distance} is intentionally secondary. It is a generic
calibration benchmark about two uniform null marginals, not an NEF
characterization. Its role is simply to show that exactness alone gives no
useful pointwise proximity; the model-specific discrepancy and decision
calculations provide the statistically relevant complement.

Several questions remain. Most notably, a complete unrestricted $n=1$
classification of ET--LR and density--LR coincidence outside the
power-variance or sample-size-stable settings would be of interest. The
quadratic-root restriction \eqref{eq:Vquadraticroot}, together with the
density identity \eqref{eq:l1} {in the global form \eqref{eq:f-at-mean} of
Proposition~\ref{prop:h-from-V}}, supplies a concrete starting point for the
unrestricted ET--LR problem. Beyond one-parameter NEFs, analogous questions
arise in multivariate exponential families, invariant testing, discrete
models with randomization, and nuisance-parameter problems. In each case the
central issue is not merely whether a $p$-value can be calibrated, but
whether distinct calibration-compatible orderings encode the same notion of
extremeness.

\appendix

\section{Symmetry and nesting of UMPU tests}

\begin{lemma}[Equal-tail unbiasedness and symmetry]
\label{lem:ET-unbiased} {Let $\mathcal{F}(\mu)$ be a one-parameter NEF
satisfying Assumption~\ref{ass:main}, fix $\theta_0\in\Theta$, and let $%
X\sim P_{\theta_0}$ be the corresponding null member, with continuous
strictly increasing CDF $F$, finite mean $m=k^{\prime }(\theta_0)$ and
quantile function $q_u=F^{-1}(u)$. Unbiasedness below is understood relative
to this family.} For $0<p<1/2$ define 
\begin{equation*}
R_p=(-\infty,q_p]\cup[q_{1-p},\infty).
\end{equation*}
In a one-parameter NEF, the level-$2p$ test $\mathbf{1}\{X\in R_p\}$ is
unbiased for the two-sided alternative if and only if 
\begin{equation}  \label{eq:ETunb}
\mathbb{E}\{X\mathbf{1}(X\in R_p)\}=2pm.
\end{equation}
Moreover, \eqref{eq:ETunb} holds for every $p\in(0,1/2)$ if and only if $F$
is symmetric about $m$.
\end{lemma}

\begin{proof}
If the level-$2p$ equal-tail test is unbiased, its power function has a
minimum at $\theta_0$. Differentiability of the NEF power function therefore
gives the necessary condition 
\begin{equation*}
0=\left.\frac{d}{d\theta}\mathbb{E}_\theta\mathbf{1}\{X\in
R_p\}\right|_{\theta=\theta_0} =\mathbb{E}_{\theta_0}\{(X-m)\mathbf{1}(X\in
R_p)\},
\end{equation*}
which is exactly \eqref{eq:ETunb}.

Conversely, the equal-tail region already has size $2p$. If \eqref{eq:ETunb}
holds, then this two-tail test satisfies both the size and first-moment
equations that characterize the bilateral UMPU test in a one-parameter
exponential family; see Lehmann and Romano \cite[Section~4.2]%
{LehmannRomano2022}. Hence the equal-tail test is UMPU and, in particular,
unbiased. This step uses the standard UMPU theorem; it does not use the
generally false implication that local unbiasedness by itself would imply
global unbiasedness.

Now write $X\overset{d}{=}q_U$ with $U\sim\mathrm{Unif}(0,1)$. Condition %
\eqref{eq:ETunb} for every $p$ becomes 
\begin{equation*}
\int_0^p q_u\,du+\int_{1-p}^1q_u\,du=2pm.
\end{equation*}
Differentiating in $p$ yields 
\begin{equation*}
q_p+q_{1-p}=2m,
\end{equation*}
which is equivalent to symmetry about $m$. The converse follows by
integration.
\end{proof}

\begin{lemma}[A symmetric NEF is Gaussian]
\label{lem:symmetric-NEF} Let $F=\{P_\theta:\theta\in\Theta\}$ be a
nondegenerate NEF with open $\Theta$. If every $P_\theta$ is symmetric about
some center $c(\theta)$, then 
\begin{equation*}
k(\theta)=\frac{\sigma^2}{2}\theta^2+a\theta+b, \qquad \sigma^2>0,
\end{equation*}
and the NEF is Gaussian up to affine transformation. The symmetry center is
necessarily 
\begin{equation*}
c(\theta)=k^{\prime }(\theta).
\end{equation*}
\end{lemma}

\begin{proof}
Because $\Theta$ is open, all moments exist locally under each $P_\theta$.
Symmetry and integrability imply that the symmetry center equals the mean: 
\begin{equation*}
c(\theta)=\mathbb{E}_\theta X=k^{\prime }(\theta).
\end{equation*}
The third central moment therefore vanishes. In an NEF, the third cumulant
is $k^{\prime \prime \prime }(\theta)$ and equals the third central moment.
Hence 
\begin{equation*}
k^{\prime \prime \prime }(\theta)=0\qquad\hbox{for every }\theta\in\Theta.
\end{equation*}
Thus $k^{\prime \prime }$ is a positive constant, and $k$ is quadratic. The
resulting moment generating function is that of a normal distribution. The
converse is immediate.
\end{proof}

\begin{proposition}[Existence, uniqueness and nesting of the continuous UMPU
critical points]
\label{prop:nesting} Fix $\theta_0$ and let $f_0$ be continuous and strictly
positive on an open support interval $S=(\ell,r)$, with finite mean $m_0\in S
$. For every $\alpha\in(0,1)$ there is a unique pair 
\begin{equation*}
\ell<c_1(\alpha)<m_0<c_2(\alpha)<r
\end{equation*}
satisfying \eqref{eq:size}--\eqref{eq:unb}. Moreover, 
\begin{equation*}
c_1^{\prime }(\alpha)>0, \qquad c_2^{\prime }(\alpha)<0,
\end{equation*}
so the UMPU rejection regions are strictly nested in $\alpha$.
\end{proposition}

\begin{proof}
It is convenient to work with the acceptance interval. For $a\in(\ell,m_0)$
define 
\begin{equation*}
J(a,b)=\int_a^b(x-m_0)f_0(x)\,dx.
\end{equation*}
At $b=m_0$, $J(a,m_0)<0$. As $b\uparrow r$, 
\begin{equation*}
J(a,r) =-\int_\ell^a(x-m_0)f_0(x)\,dx>0,
\end{equation*}
because the total centered first moment is zero. Since $J(a,b)$ is strictly
increasing in $b>m_0$, there is a unique $b=b(a)>m_0$ with $J(a,b(a))=0$.

Differentiate the identity $J(a,b(a))=0$: 
\begin{equation*}
(b-m_0)f_0(b)b^{\prime }(a)-(a-m_0)f_0(a)=0,
\end{equation*}
so 
\begin{equation*}
b^{\prime }(a)=\frac{(a-m_0)f_0(a)}{(b-m_0)f_0(b)}<0.
\end{equation*}
Let 
\begin{equation*}
A(a)=F_0(b(a))-F_0(a)
\end{equation*}
be the acceptance probability. Then 
\begin{equation*}
A^{\prime }(a) =f_0(b)b^{\prime }(a)-f_0(a) =f_0(a)\frac{a-b}{b-m_0}<0.
\end{equation*}
As $a\downarrow\ell$, one has $b(a)\uparrow r$ and $A(a)\uparrow1$; as $%
a\uparrow m_0$, one has $b(a)\downarrow m_0$ and $A(a)\downarrow0$. Hence
for each $\alpha$ there is a unique $a=c_1(\alpha)$ with $A(a)=1-\alpha$,
and $c_2(\alpha)=b(c_1(\alpha))$.

Since $A^{\prime }(a)<0$, 
\begin{equation*}
c_1^{\prime }(\alpha)=-\frac{1}{A^{\prime }(c_1(\alpha))}>0.
\end{equation*}
Because $b^{\prime }(a)<0$, 
\begin{equation*}
c_2^{\prime }(\alpha)=b^{\prime }(c_1(\alpha))c_1^{\prime }(\alpha)<0.
\end{equation*}
Thus increasing the level expands both tails and the rejection regions are
nested.
\end{proof}

\begin{lemma}[Symmetric nulls give equal-tail UMPU regions]
\label{lem:sym-UMPU} If, in Proposition~\ref{prop:nesting}, the null
distribution is symmetric about $m_0$, then 
\begin{equation*}
c_2(\alpha)-m_0=m_0-c_1(\alpha),
\end{equation*}
and each UMPU rejection region is the equal-tail rejection region.
\end{lemma}

\begin{proof}
Reflection $x\mapsto2m_0-x$ preserves the null law. If $(c_1,c_2) $
satisfies the size and unbiasedness equations, then so does $%
(2m_0-c_2,2m_0-c_1)$. Uniqueness in Proposition~\ref{prop:nesting} therefore
forces 
\begin{equation*}
(c_1,c_2)=(2m_0-c_2,2m_0-c_1).
\end{equation*}
Symmetry then gives equal null tail probabilities.
\end{proof}

\section{Positive stable generating laws}

{\ This appendix records only the stable-law facts needed in the paper. For
standard treatments see Feller \cite{Feller1971}, Zolotarev \cite%
{Zolotarev1986}, and Nolan \cite{Nolan2020}. Explicit closed forms for many
rational stability indices are given by Penson and G\'orska \cite%
{PensonGorska2010}. }

Fix $0<\alpha<1$. Let $g_\alpha$ be the positive one-sided $\alpha$-stable
density determined by 
\begin{equation}  \label{eq:stableLT}
\int_0^\infty e^{-s x}g_\alpha(x)\,dx=e^{-s^\alpha}, \qquad s>0.
\end{equation}
Equivalently, with $\theta=-s<0$, 
\begin{equation*}
\int_0^\infty e^{\theta x}g_\alpha(x)\,dx =\exp\{-(-\theta)^\alpha\}.
\end{equation*}
The density is positive and smooth on $(0,\infty)$. The NEF generated by $%
g_\alpha(x)dx$ is 
\begin{equation}  \label{eq:stableNEF}
f_{\alpha,\theta}(x) =\exp\{\theta x+(-\theta)^\alpha\}g_\alpha(x), \qquad
x>0,\quad\theta<0.
\end{equation}
Its cumulant transform is 
\begin{equation*}
k(\theta)=-(-\theta)^\alpha,
\end{equation*}
so 
\begin{equation*}
m(\theta)=\alpha(-\theta)^{\alpha-1},
\end{equation*}
and 
\begin{equation*}
V(m) =(1-\alpha)\alpha^{-1/(1-\alpha)} m^{(2-\alpha)/(1-\alpha)}.
\end{equation*}
The mean domain is $(0,\infty)$, which equals the interior of the convex
support, and therefore the family is steep. These power-variance NEFs belong
naturally to the class studied by Bar-Lev and Enis \cite{BarLevEnis1986}.

At $\alpha=1/2$ one obtains the explicit L\'evy density \eqref{eq:Levy}, and
exponential tilting gives the inverse-Gaussian NEF used in Section~\ref%
{sec:examples}. For $\alpha\neq1/2$, the corresponding positive stable NEFs
provide smooth steep examples outside the normal--gamma--inverse-Gaussian
coincidence triple of \cite{BarLevBshoutyLetac2002}.

\section{Derivative calculation for the ET--LR variance-function equation}

\label{app:ETLRcalc}

{\ We record the algebra used in Theorem~\ref{thm:ETLR-ODE} in a form that
can be checked directly. Fix a mean value $m$ and suppress its argument.
Write 
\begin{equation*}
f=f_{\psi(m)}(m), \qquad \ell_j= \left. \frac{\partial^j}{\partial x^j}\log
f_{\psi(m)}(x) \right|_{x=m}.
\end{equation*}
For the log LR function $r_m$ one has 
\begin{equation*}
r_1:=r_m^{\prime }(m)=0, \qquad r_2:=r_m^{\prime \prime }(m)=-\frac1V,
\qquad r_3:=r_m^{\prime \prime \prime }(m)=\frac{V^{\prime }}{V^2},
\end{equation*}
\begin{equation*}
r_4:=r_m^{(4)}(m) =\frac{V^{\prime \prime }}{V^2}-\frac{2(V^{\prime })^2}{V^3%
},
\end{equation*}
and 
\begin{equation*}
r_5:=r_m^{(5)}(m) =\frac{V^{\prime \prime \prime }}{V^2} -\frac{6V^{\prime
}V^{\prime \prime }}{V^3} +\frac{6(V^{\prime })^3}{V^4}.
\end{equation*}
If $q=F^{-1}$ and $G=r_m\circ q$, differentiation in the quantile coordinate
can be written as 
\begin{equation*}
\frac{d}{du} =\frac{1}{f(q(u))}\frac{d}{dx}.
\end{equation*}
At $u=1/2$, ET--LR equality gives $q(1/2)=m$. Since $r_1=0$, direct
differentiation gives 
\begin{equation*}
G^{\prime \prime \prime }(1/2) =\frac{1}{f^3}\{r_3-3\ell_1r_2\}.
\end{equation*}
The symmetry $G(u)=G(1-u)$ forces $G^{\prime \prime \prime }(1/2)=0$, and
therefore 
\begin{equation}  \label{eq:app-l1}
\ell_1=-\frac{V^{\prime }}{3V}.
\end{equation}
Because 
\begin{equation*}
\ell_1(m)=\psi(m)+(\log h)^{\prime }(m), \qquad \psi^{\prime }(m)=\frac1V,
\end{equation*}
differentiating \eqref{eq:app-l1} with respect to $m$ gives 
\begin{equation*}
\ell_2 =-\frac1V-\frac{V^{\prime \prime }}{3V}+\frac{(V^{\prime })^2}{3V^2},
\end{equation*}
and one further differentiation gives 
\begin{equation*}
\ell_3 =\frac{V^{\prime }}{V^2} -\frac{V^{\prime \prime \prime }}{3V} +\frac{%
V^{\prime }V^{\prime \prime }}{V^2} -\frac{2(V^{\prime })^3}{3V^3}.
\end{equation*}
}

The fifth derivative can be organized without suppressing the chain-rule
terms. At $u=1/2$, 
\begin{align}
G^{(5)}(1/2) =\frac1{f^5}\big[\,&r_5-10\ell_1r_4 +(35\ell_1^2-10\ell_2)r_3 \\
&+(-5\ell_3+45\ell_1\ell_2-50\ell_1^3)r_2\big].  \label{eq:G5expanded}
\end{align}
Substituting the displayed expressions for $r_2,r_3,r_4,r_5$ and $%
\ell_1,\ell_2,\ell_3$ into \eqref{eq:G5expanded} and collecting terms gives 
\begin{equation*}
G^{(5)}(1/2) =-\frac{2}{27f^5V^4} \left\{ 9V^2V^{\prime \prime \prime
}-9VV^{\prime }V^{\prime \prime }+4(V^{\prime })^3\right\}.
\end{equation*}
Since $G$ is symmetric about $1/2$, $G^{(5)}(1/2)=0$, yielding %
\eqref{eq:ETLR-ODE}.

For completeness, direct differentiation of $V^{2/3}$ gives 
\begin{equation*}
\left(V^{2/3}\right)^{\prime \prime \prime }= \frac{2}{27V^{7/3}} \left\{
9V^2V^{\prime \prime \prime }-9VV^{\prime }V^{\prime \prime }+4(V^{\prime
})^3\right\},
\end{equation*}
so the differential equation is exactly equivalent to $(V^{2/3})^{\prime
\prime \prime }=0$.

{\ }


\begin{thebibliography}{99}
\bibitem{Wasserstein2016} R. L. Wasserstein and N. A. Lazar, The ASA
statement on p-values: context, process, and purpose, \emph{The American
Statistician} \textbf{70} (2016), no. 2, 129--133. DOI:
10.1080/00031305.2016.1154108.

\bibitem{Wasserstein2019} R. L. Wasserstein, A. L. Schirm and N. A. Lazar,
Moving to a world beyond $p<0.05$, \emph{The American Statistician} \textbf{%
73} (2019), sup1, 1--19. DOI: 10.1080/00031305.2019.1583913.

\bibitem{LehmannRomano2022} E. L. Lehmann and J. P. Romano, \emph{Testing
Statistical Hypotheses}, 4th ed., Springer, Cham, 2022. DOI:
10.1007/978-3-030-70578-7.

\bibitem{Birkes1990} D. Birkes, Generalized likelihood ratio tests and
uniformly most powerful tests, \emph{The American Statistician} \textbf{44}
(1990), no. 2, 163--166. DOI: 10.1080/00031305.1990.10475708.

\bibitem{BergerDelampady1987} J. O. Berger and M. Delampady, Testing precise
hypotheses, \emph{Statistical Science} \textbf{2} (1987), no. 3, 317--335.
DOI: 10.1214/ss/1177013238.

\bibitem{Kulinskaya2008} E. Kulinskaya, On two-sided p-values for
non-symmetric distributions, arXiv:0810.2124, 2008.

\bibitem{GibbonsPratt1975} J. D. Gibbons and J. W. Pratt, P-values:
interpretation and methodology, \emph{The American Statistician} \textbf{29}
(1975), no. 1, 20--25. DOI: 10.1080/00031305.1975.10479106.

\bibitem{GeorgeMudholkar1990} E. O. George and G. S. Mudholkar, P-values for
two-sided tests, \emph{Biometrical Journal} \textbf{32} (1990), no. 6,
747--751. DOI: 10.1002/bimj.4710320615.

\bibitem{DunnePawitanDoody1996} A. Dunne, Y. Pawitan and L. Doody, Two-sided
p-values from discrete asymmetric distributions based on uniformly most
powerful unbiased tests, \emph{Journal of the Royal Statistical Society.
Series D (The Statistician)} \textbf{45} (1996), no. 4, 397--405. DOI:
10.2307/2988542.

\bibitem{MudholkarChaubey2009} G. S. Mudholkar and Y. P. Chaubey, On
defining P-values, \emph{Statistics \& Probability Letters} \textbf{79}
(2009), no. 18, 1963--1971. DOI: 10.1016/j.spl.2009.06.006.

\bibitem{BarLevBshoutyLetac2002} S. K. Bar-Lev, D. Bshouty and G. Letac,
Normal, gamma and inverse-Gaussian are the only NEFs where the bilateral
UMPU and GLR tests coincide, \emph{The Annals of Statistics} \textbf{30}
(2002), no. 5, 1524--1534. DOI: 10.1214/aos/1035844987.

\bibitem{BarLevKokonendji2017} S. K. Bar-Lev and C. C. Kokonendji, On the
mean value parametrization of natural exponential families---a revisited
review, \emph{Mathematical Methods of Statistics} \textbf{26} (2017), no. 3,
159--175. DOI: 10.3103/S1066530717030012.

\bibitem{LetacMora1990} G. Letac and M. Mora, Natural real exponential
families with cubic variance functions, \emph{The Annals of Statistics} 
\textbf{18} (1990), no. 1, 1--37. DOI: 10.1214/aos/1176347491.

\bibitem{BarndorffNielsen1997} O. E. Barndorff-Nielsen, Normal inverse
Gaussian distributions and stochastic volatility modelling, \emph{%
Scandinavian Journal of Statistics} \textbf{24} (1997), no. 1, 1--13. DOI:
10.1111/1467-9469.00045.

\bibitem{BarndorffNielsen1978} O. E. Barndorff-Nielsen, \emph{Information
and Exponential Families in Statistical Theory}, John Wiley \& Sons,
Chichester, 1978.

\bibitem{BarLevEnis1986} S. K. Bar-Lev and P. Enis, Reproducibility and
natural exponential families with power variance functions, \emph{The Annals
of Statistics} \textbf{14} (1986), no. 4, 1507--1522. DOI:
10.1214/aos/1176350173.

\bibitem{Fischer2014} M. J. Fischer, \emph{Generalized Hyperbolic Secant
Distributions: With Applications to Finance}, SpringerBriefs in Statistics,
Springer, Berlin--Heidelberg, 2014. DOI: 10.1007/978-3-642-45138-6.

\bibitem{Feller1971} W. Feller, \emph{An Introduction to Probability Theory
and Its Applications, Vol. II}, 2nd ed., Wiley, 1971.

\bibitem{Zolotarev1986} V. M. Zolotarev, \emph{One-Dimensional Stable
Distributions}, American Mathematical Society, 1986.

\bibitem{Nolan2020} J. P. Nolan, \emph{Univariate Stable Distributions:
Models for Heavy Tailed Data}, Springer Series in Operations Research and
Financial Engineering, Springer, Cham, 2020. DOI: 10.1007/978-3-030-52915-4.

\bibitem{PensonGorska2010} K. A. Penson and K. Gorska, Exact and explicit
probability densities for one-sided Levy stable distributions, \emph{%
Physical Review Letters} \textbf{105} (2010), 210604. DOI:
10.1103/PhysRevLett.105.210604.

\bibitem{LetacMattnerPiccioni2018} G. Letac, L. Mattner and M. Piccioni, The
median of an exponential family and the normal law, \emph{Statistics \&
Probability Letters} \textbf{133} (2018), 38--41. DOI:
10.1016/j.spl.2017.10.002.

\bibitem{PiccioniKolodziejekLetac2020} M. Piccioni, B. Ko\l odziejek and G.
Letac, Location and scale behaviour of the quantiles of a natural
exponential family, \emph{ESAIM: Probability and Statistics} \textbf{24}
(2020), 244--251. DOI: 10.1051/ps/2019009.

\bibitem{Jorgensen1987} B. J\o rgensen, Exponential dispersion models, \emph{%
Journal of the Royal Statistical Society, Series B} \textbf{49} (1987), no.
2, 127--145. DOI: 10.1111/j.2517-6161.1987.tb01685.x.

\bibitem{Petrov1975} V. V. Petrov, \emph{Sums of Independent Random Variables%
}, Ergebnisse der Mathematik und ihrer Grenzgebiete, vol. 82,
Springer-Verlag, Berlin--Heidelberg, 1975. DOI: 10.1007/978-3-642-65809-9.

\bibitem{RCore2026} R Core Team, \emph{R: A Language and Environment for
Statistical Computing}, R Foundation for Statistical Computing, Vienna,
Austria, 2026.
\end{thebibliography}
\end{document}